\documentclass[12pt]{article}
\pdfoutput=1

\usepackage{physics} 
\usepackage{siunitx} 
\usepackage{enumerate} 
\usepackage{pgfplots}
\usepackage{algorithm}
\usepackage{bm}
\usepackage{graphicx}	
\usepackage{algpseudocode}
\usepackage{pgfplotstable}
\usepackage{tikz,pgfplots}
\usepackage{subcaption}
\usepackage{hyperref}
\usepackage{amsmath, amssymb}
\usepackage{amsthm}
\usepackage{appendix}
\usepackage{geometry}
\usepackage{enumitem}
\usepackage{dsfont}
\usepackage{listings}
\usepackage{xcolor}

\lstdefinestyle{custompython}{
  language=Python,
  basicstyle=\ttfamily\footnotesize,
  commentstyle=\color{green!40!black},
  keywordstyle=\color{blue},
  numberstyle=\tiny\color{gray},
  numbers=left,
  numbersep=5pt,
  backgroundcolor=\color{gray!5},
  showspaces=false,
  showstringspaces=false,
  showtabs=false,
  frame=single,
  tabsize=2,
  breaklines=true,
  breakatwhitespace=false,
  captionpos=b,
  belowcaptionskip=10pt,
  xleftmargin=15pt,
  xrightmargin=15pt,
}

\newtheorem{proposition}{Proposition}
\newtheorem{definition}{Definition}
\newtheorem{theorem}{Theorem}
\newtheorem{corollary}{Corollary}

\pgfplotsset{compat=1.14}

\begin{document}

\title{Analysis of Hopfield Model as Associative Memory} 
\author{Silvestri Matteo} 

\date{Department of Mathematics, University of Rome “La Sapienza”}

\maketitle  

\begin{abstract} 
	This article delves into the Hopfield neural network model, drawing inspiration from biological neural systems. The exploration begins with an overview of the model's foundations, incorporating insights from mechanical statistics to deepen our understanding. Focusing on \emph{audio retrieval}, the study demonstrates the Hopfield model's associative memory capabilities. Through practical implementation, the network is trained to retrieve different patterns. 

\end{abstract}

\section{An Overview of Neuroscience}
	 
The Hopfield model finds inspiration in the intricate connections among biological neurons in the human brain, mimicking nature's efficiency in information processing. Similar to the synaptic communication in biological systems, the Hopfield network utilizes connections to store and retrieve patterns.  Before delving further, we'll provide a brief overview of biological neurons, shedding light on their fundamental structures. 

\subsection{Biological Neurons}
A neuron, the fundamental building block of our nervous system, consists of three main parts: the dendrites, the cell body (soma), and the axon. \textbf{Dendrites} receive signals from other neurons, transmitting these signals to the cell body. The \textbf{Cell body} processes these signals and, if the input is sufficient, generates an electrical impulse. This impulse travels down the \textbf{Axon}, a long, slender projection, to communicate with other neurons or muscles. \textbf{Synapses}, the junctions between neurons, facilitate this communication by transmitting electrochemical signals to other cells. This intricate architecture enables neurons to form complex networks, laying the foundation for the remarkable functionality of our nervous system.

\begin{figure}[htpb]
    \centering
    \begin{subfigure}{0.45\textwidth}
        \centering
        \includegraphics[width=\linewidth]{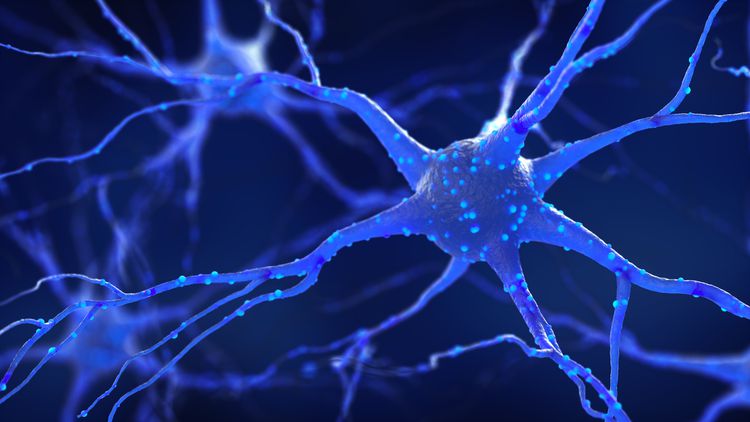}
    \end{subfigure}
    \hfill
    \begin{subfigure}{0.45\textwidth}
        \centering
        \includegraphics[width=\linewidth]{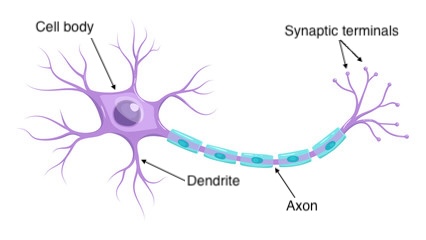}
    \end{subfigure}
    \caption{On the left a 3D-image of a neuron, on the right a brief summary of its structure.}
    \label{fig1}
\end{figure}

Embarking on the marvels of neural networks, the sheer scale of these intricate systems is staggering. In the human brain alone, an astounding \emph{86 billion neurons} form a vast web of connections, orchestrating the symphony of thoughts and actions. Comparatively, the cognitive prowess of elephants shines through their expansive neural landscapes, boasting approximately 257 billion neurons—more than three times that of humans. Dolphins, renowned for their intelligence, navigate the seas with brains equipped with tens billion of neurons, contributing to their advanced problem-solving abilities. These examples illuminate the remarkable diversity and complexity of neural networks across species, and moreover show the huge quantity of neurons that are used in daily actions. In the forthcoming section, we delve into a concise overview of a neuron's behavior, unraveling the fundamental processes that underlie its intricate functioning. 

\subsection{Action Potential}
While our focus is on the Hopfield network model, a brief foray into the neuroscience background is essential. In this section, we won't delve into the depths of neuroscience, but rather aim to illuminate a fundamental concept: \emph{action potential}. Understanding the behavior of neurons and the process that gives rise to a spike lays a crucial foundation. This rudimentary insight serves as a key building block, enriching our comprehension of the Neural Network models and its mathematical underpinnings. The action potential $\mathcal{AP}$, a pivotal concept in neuronal function,
is a neuronal phenomenon in which we see the neuron fires. Transmission of a neuronal signal is entirely dependent of the movement of ions, such as Sodium (Na+), Potassium (K+) and Chloride (CI), that are unequally distributed between the inside and the outside of the cell body. The presence and the movement of these ions creates a chemical gradient across the membrane which we define as electro-chemical gradient $\mathcal{ECG}$.
 \vspace{\baselineskip}

\begin{figure}[htpb]
    \centering
    \includegraphics[width=0.7\textwidth]{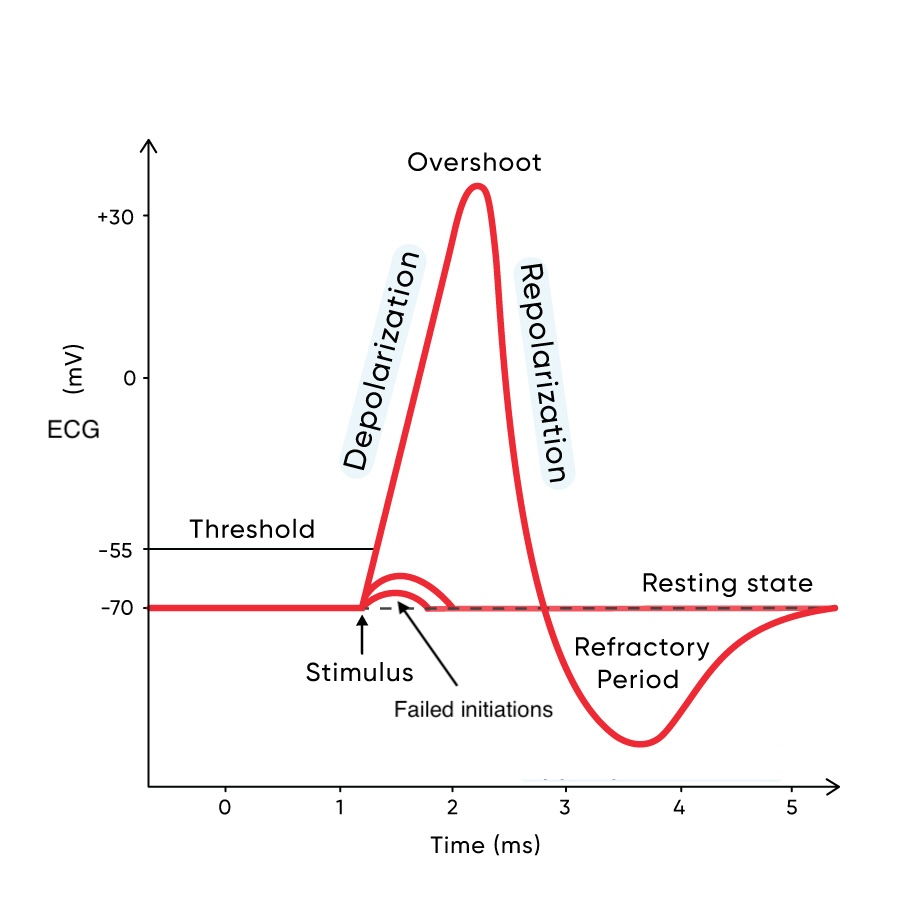}
    \caption{Action Potential}
    \label{fig2}
\end{figure}
At resting state, $\mathcal{ECG}$ hovers around $-70 mV$. However, when a neuron receives a stimulus, the action potential experiences a tendency to increase. Within the biological neuron, a critical threshold exists, typically around $-55 mV$. If ECG exceeds this threshold, then the neuron activates and the process of generating a nerve signal begins. Otherwise, the neuron is unable to fire and tends to return to its resting state. Let's focus on the case where the \textbf{Stimulus} is strong enough to cause ECG to exceed threshold. Then the neuron activates, and the \textbf{Depolarization} process begins. In this state, the neuron begins to interact with ions present inside and outside the membrane in such a way that ECG continues to grow up to $+40 mV$; this is called an overshoot. At this point, the membrane begins to expel positive ions in order to do ECG decrease; this process is called \textbf{Repolarization}. Following these processes, the neuron is able to generate an electrical signal that is sent through the axon to reach the target cell. In more detail, after the depolarization process there is the so-called \textbf{Refractory period}. During this time segment, ECG drops below the resting-state value. This happens because the channels present in the membrane that allow the ions to cross it do not close instantaneously and therefore allow values smaller than $-70mV$ to be reached. The neuron subsequently restabilizes and returns to a resting state. As we can see in figure \ref{fig2} the trend of the $\mathcal{ECG}$ takes on a shape of a spike and allows us to imagine the production of an electrical signal.
Navigating the intricacies of neuroscience, especially outside one's specialization, can be challenging. To enhance clarity, I've included a link to a \href{https://www.youtube.com/watch?v=oa6rvUJlg7o}{video explanation}.
However, two basic concepts on which the associated mathematical models are based should be clear:
\begin{itemize}
    \item Cognitive capacity does not depend on any intensity, but only on \emph{binary values} (and more specifically, by frequency).
    \item There should exists a \emph{threshold} of the network that allows to activate or not neurons.
\end{itemize}

\section{Artificial Neurons}
As we transition into the realm of "Artificial Neurons", our focus shifts from the intricate workings of biological neurons to their mathematical counterparts. These artificial neurons serve as the foundational units in computational models, mirroring the neurological properties we've explored in the preceding section. Embodying the essence of their biological counterparts, artificial neurons encapsulate key features like activation thresholds and the generation of binary outputs, all within a mathematical framework. This section unravels the fundamental principles behind these artificial neurons, bridging the gap between neuroscience and mathematical modeling in the pursuit of understanding neural networks.
    \subsection{McCulloch-Pitts Model}
    The McCulloch-Pitts model, known as \textbf{MP Neuron}, stands as the epitome of simplicity in neural network modeling. Comprising inputs, weights, and a threshold, this foundational model captures the essence of a neuron's basic functionality. Inputs convey signals, each associated with a weight, which collectively influence the neuron's behavior. The critical threshold, akin to the activation threshold in biological neurons, determines whether the neuron fires or remains at rest. As we can see from the figure \ref{fig3}, the model can be briefly summarized in the equation
    \begin{equation}
        y = \Theta \left( \sum_{k=1}^{N} J_{k} S_{k} - U^{*} \right)
    \end{equation}
    \begin{figure}[htpb]
    \centering
    \includegraphics[width=0.7\textwidth]{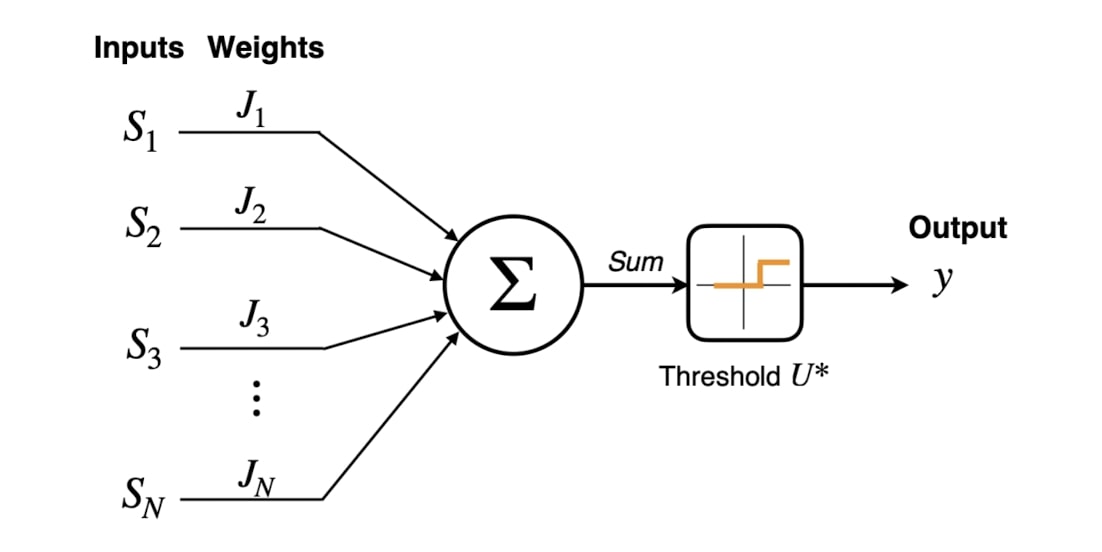}
    \caption{MP model}
    \label{fig3}
\end{figure}
    where $\Theta$ is the \emph{Heaviside function}, $N$ is  number of inputs $S_{k}$, $J_{k}$ are the synaptic weights and $U^{*}$ is the neuron threshold.
    This general model obviously embodies the main behavioral characteristics of biological neurons. More specifically, a network of MP neurons can perform any single-out mapping $M : \Omega \subset \{0,1\}^{N} \rightarrow \{0,1\}$; for instance, any boolean function of $N$ variables, can be expressed in terms of AND,OR,NOT operations $(\land,\lor,\lnot)$. These function can be easily implemented with a network of two MP neurons (only one in the case of $\lnot$). However, if N is greater than 1, there is a counterexample that shows how a single layer of MP neurons is not sufficient to approximate any objective function M : XOR operation. For this reason, scientific attention has shifted to
    \textbf{MP Multilayer Networks}. Indeed, the XOR function can be performed starting from a neural network composed of two layers, each with two MP neurons. If now we assume that we don't know the function $M$ but we have a \emph{training set} $TS = \{x_{i},M(x_{i}) \}_{i=1}^{P} \subset \Omega \times \{0,1\}$, then there exists an algorithm knows as \textbf{Perceptron} that allow us to train an MP neuron in order to emulate as best as possible the target value M($\cdot$):

    \begin{algorithm}
    \caption{Perceptron Learning}
    \begin{algorithmic}[1]
        \State Define $y = y(x) = \Theta( J \cdot x - U^{*} )$ with randomly chosen parameters $\{J_{k}\}_{k=1}^{N}$ and $U^{*}$
        \State $i = 1$
        \While{termination condition not reached}
            \State \textbf{if} $M(x_i) = y(x_i)$ \textbf{then} keep going
            \State \textbf{if} $M(x_i)=0 \land y(x_i) = 1$ \textbf{then} $U^{*} = U^{*}+1$ , $J = J - x_i$
            \State \textbf{if} $M(x_i)=1 \land y(x_i) = 0$ \textbf{then} $U^{*} = U^{*}-1$ , $J = J + x_i$
            \State $i = i+1$
        \EndWhile
    \end{algorithmic}
    \end{algorithm}
    A possible termination condition is, for example, $y(x_i)=M(x_i) \quad \forall x_i \in TS$. \\
    Therefore the algorithm ends successfully only if $\Omega$ is \emph{linearly separable}; if not, there are methods that transform $\Omega$ into a linearly separable space in order to train the neuron correctly.

\subsection{Neural Networks}
\label{2.2}
    From now on, we want to study the model of neural networks, in which there is mutual information between inputs and outputs. Let's consider a network with $N$ neurons $S_1 ... S_N$; we denote with $J_{ij}$ the synaptic weight between the neurons $S_i$ and $S_j$.
    Since the network is no longer feed-forward, we are interested in expressing the state of each individual neuron as a function of time $t$. More specifically, if we assume that we know the neuron states $S_1(t) ... S_N(t)$ then the behavior of $i$-th neuron can be expressed by the equation
    \begin{equation}
        S_i(t + \Delta t ) = \Theta \left( \sum_{k=1}^{N} J_{ik}  S_k(t) - U_i^{*}\right)
    \end{equation}
    where $U_i^{*}$ is the threshold of the neuron $S_i$. However, this modeling turns out to be a bit unrealistic, as it does not take into account the fact that the neuron's threshold could vary over time. Indeed, we consider the \textbf{Stochastic Neurons} 
    \begin{equation}
    \left\{
    \begin{aligned}
        U_i^{*} (t) &= U_i^{*} - \dfrac{T}{2} z_i(t) \\
        S_i(t + &\Delta t ) = \Theta \left( \sum_{k=1}^{N} J_{ik}  S_k(t) - U_i^{*}(t) \right)
    \end{aligned}
    \right.
    \end{equation}
    with the \emph{noise term} such that $\mathbb{E}\left(z_i(t)\right) = 0$ and $\mathbb{E}\left(z_i(t)^2\right) = 1$; the \emph{temperature} $T$ is a very important control parameter, which plays a central role in the computational and convergence properties of the Hopfield model. A second convenient traslation is to redefine neurons such that they have values in $\{+1,-1\}$, so called \textbf{Ising Neurons}
    \begin{equation}
    \left\{
    \begin{aligned}
        \sigma_i&(t) = 2 S_i(t) - 1 \\
        U_i^{*} &= \frac{1}{2} \left( \sum_{k=1}^N J_{ik} - h_i \right)
    \end{aligned}
    \right.
    \end{equation}
    where $\sigma_i(t) \in \{-1,+1\}$ and $\{h_i\}_{i=1}^N$ are the biases. Let's denote the \emph{local field} acting on the neuron $\sigma_i$ as  
    \begin{equation*}
        \varphi_i(t) := \sum_{k=1}^N J_{ik} \sigma_k(t) + h_i
    \end{equation*}
    After some simple calculations, we obtain the generic formula
    \begin{equation}
    \label{eq5}
        \sigma_i(t + \Delta t ) = \operatorname{sgn}\left(\varphi_i(t) + T z_i(t)\right) 
    \end{equation}
    with $\operatorname{sgn}(x) = 2 \Theta(x) -1$; this will be the notation we will use as the article continues. The probability to find a neuron state $\sigma_i(t + \Delta t )$ can be expressed in terms of the noise distribution $\mathcal{P}(z)$; in the case of \emph{symmetric distribution}, we have
    \begin{equation}
    \label{eq6}
        \mathbb{P} \left( \sigma_i(t + \Delta t ) = \pm 1 \right) = g \left( \pm \dfrac{\varphi_i(t)}{T} \right) \equiv 
        \int_{-\infty}^{\pm \dfrac{\varphi_i(t)}{T}}
        \mathcal{P}(z)\,dz
    \end{equation}
    where $g$ is the \emph{cumulative distribution function}. A natural choice is to consider the distribution of a Standard Gaussian, whose associated \emph{cdf} is $g(x) = \frac{1}{2} ( 1 + \operatorname{erf}(\dfrac{x}{\sqrt{2}}))$. Another plausible choice is 
    \begin{equation}
    \left\{
    \begin{aligned}
        \mathcal{P}(z) &= \frac{1}{2} ( 1 - \operatorname{tanh}^2(z)) \\
        g(z) &= \frac{1}{2} ( 1 + \operatorname{tanh}(z) )
    \end{aligned}
    \right.
    \end{equation}
    \begin{figure}[htpb]
    \centering
    \includegraphics[width=0.4\textwidth]{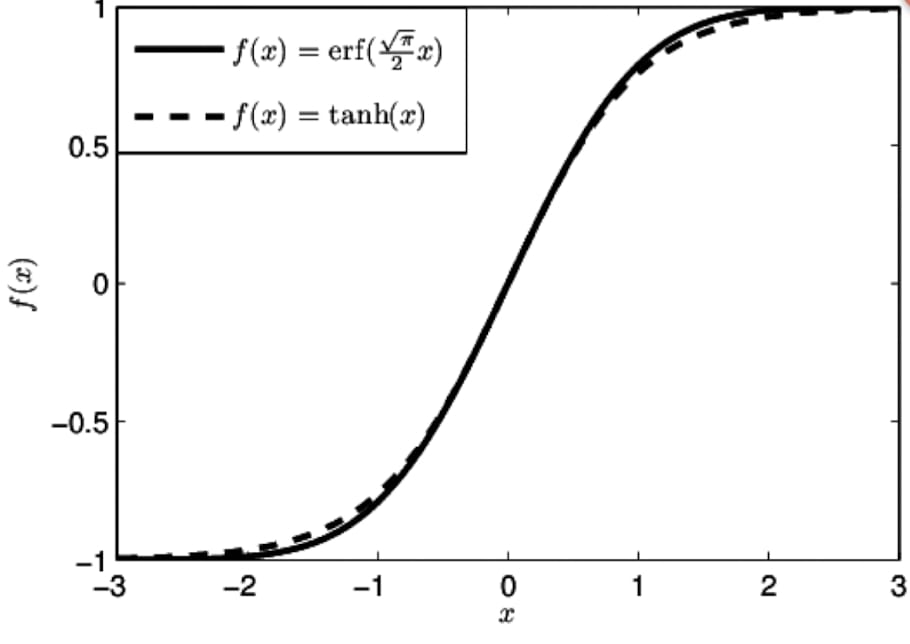}
    \caption{Comparison between ${cdf}_s$}
    \label{fig4}
    \end{figure}
    As we can see from the figure \ref{fig4}, the two possible choices are very similar on a numerical level. Notice that $T$  controls the impact of the noise on the model; in fact, if $T=0$ then the process is deterministic, while if $T \rightarrow \infty$ then  \begin{equation*}
        \mathbb{P} \left( \sigma_i(t + \Delta t ) = \pm 1 \right) = g (0) = \frac{1}{2}
    \end{equation*}
    that is, the process is fully-random. If we assume $T \neq 0$, we can see from equation \ref{eq6} that the microscopic laws governing the \emph{spin vector} $ \bm{\sigma} = (\sigma_1 ... \sigma_N)$ are defined as a stochastic alignment to the local field $\bm{\varphi} = (\varphi_1(\bm{\sigma}) ... \varphi_N(\bm{\sigma}))$; as a matter of fact, if $\varphi_1(\bm{\sigma}) > 0$ then $\mathbb{P} \left( \sigma_1(t + \Delta t ) =  1 \right) > \frac{1}{2} $, while if $\varphi_1(\bm{\sigma}) < 0$ then $\mathbb{P} \left( \sigma_1(t + \Delta t ) =  1 \right) < \frac{1}{2}$.

    \subsection{Noiseless Networks}
    Now we focus on noiseless dynamics, which can be divided into two types: 
    \begin{itemize}
    \item \textbf{Parallel dynamics}, represented by 
    \begin{equation}
        \sigma_i(t + \Delta t ) = \operatorname{sgn}\left(\sum_{k=1}^N J_{ik} \sigma_k(t) + h_i\right) \quad \forall i \in \{1..N\}
    \end{equation}
    \item \textbf{Sequential dynamics}, represented by
        \begin{equation}
         \label{eq9}
        \left\{
        \begin{aligned}
            &\text{choose randomly  } i \text{  in } \{1..N\} \\
            &\sigma_i(t + \Delta t ) = \operatorname{sgn}\left(\sum_{k=1}^N J_{ik} \sigma_k(t) + h_i\right)
        \end{aligned}
        \right.
        \end{equation}
    \end{itemize}
    Therefore, starting from an initial configuration $\bm{\sigma}(0)$, the following sequence is obtained
    \begin{equation}
    \label{eq10}
        \bm{\sigma}(0) \rightarrow \bm{\sigma}(1) \rightarrow \bm{\sigma}(2) \rightarrow \text{...}
    \end{equation}
    and we are hopeful that the sequence tends towards an \emph{attractor} $\bm{\sigma^{*}}$ or in a limit cycle. Parallel dynamics and sequential dynamics (with some improvements) always evolve into an attractor, that is a limit cycle of period less than or equal to $2^N$. However, we will focus on sequential dynamics since it is evident that, for very large $N$, parallel dynamics turns out to be very expensive from a computational point of view.

    \begin{proposition}
        Let's consider a noiseless sequential dynamic that has the following properties:
        \begin{enumerate}
        [label=(\roman*)]
        \item Symmetric interactions, i.e. $J_{ik} = J_{ki} \ \ \forall i,k \in \{1...N\}$
        \item Non-negative self-interactions, i.e. $J_{ii} \geq 0 \ \forall i \in \{1...N\}$
        \item Stationary external field $\bm{h} = (h_1,...,h_N)$
    \end{enumerate}
    Then the function 
    \begin{equation}
    \label{eq11}
        L(\bm{\sigma}; N, J, \bm{h}) := -\frac{1}{2} \sum_{k,l=1}^N \sigma_k J_{kl} \sigma_j - \sum_{k=1}^N h_k \sigma_k
    \end{equation}
    is a \textbf{Ljapunov function} with respect to the dynamic written above.
    \end{proposition}
    \begin{proof}
        We have to show that $\Delta L(\bm{\sigma}) = L(\bm{\sigma'}) - L(\bm{\sigma}) \leq 0 \ \forall \bm{\sigma},\bm{\sigma'}$, where the two configurations can only be different on the state of neuron i. Wlog, we can consider $\bm{\sigma'} \neq \bm{\sigma}$ with $\sigma_i' = - \sigma_i$. Then we have
        \begin{equation*}
            \begin{aligned}
              \Delta L(\bm{\sigma}) &= -\frac{1}{2} \sum_{k=1}^N J_{ki} (\sigma'_k \sigma'_i - \sigma_k \sigma_i) -\frac{1}{2} \sum_{l=1}^N J_{il} (\sigma'_i \sigma'_l - \sigma_i \sigma_l) - h_i(\sigma'_i - \sigma_i) \\
              &= \sum_{k=1}^N J_{ki} \sigma_k \sigma_i + \sum_{l=1}^N J_{il} \sigma_i \sigma_l - 2 J_{ii} + 2 h_i \sigma_i \\
              &= 2 \sigma_i \left(\sum_{k=1}^N J_{ki} \sigma_k  + h_i\right) - 2 J_{ii} = -2 \bigg| \sum_{k=1}^N J_{ki} \sigma_k  \bigg| - 2 J_{ii} \leq 0
            \end{aligned}
        \end{equation*}  
        where we used the fact that $\psi \operatorname{sgn}(\psi) = |\psi|$ with $\psi$ generic function; in our case, we have $\psi = \sum_{k=1}^N J_{ki} \sigma_k + h_i$ and $\operatorname{sgn}(\psi) = \sigma_i' = -\sigma_i$ hence the thesis.
    \end{proof}

    The result just demonstrated shows a key concept that underlies
    spin-glass models of neural networks and their application in patterns recognition: we want an \emph{energy function} which has fixed points (previously called attractors) corresponding to the patterns we want the network to memorize. In the case of a \textbf{single Pattern} $\bm{\xi} = (\xi_1,...,\xi_N) \in \{-1,+1\}^N$, we define the synaptic weights according to the \emph{Hebbian rule}
    \begin{equation}
    	J_{ij} = \dfrac{\xi_{i} \xi_{j}}{N} \quad  \forall i,j = 1...N
    \end{equation}
    which is connected to the concept  ``cells that fire together they wire together''. Moreover, we can assume that $h_{i} = h \ \forall i$,  because there is no reason why different neurons should be feel different external fields. Let's introduce the new variables $\tau_{i } = \xi_{i} \sigma_{i}$ and $\nu_{i} = \xi_{i} h$; multiplying the equation \ref{eq9} by $\xi_{i}$ we obtain the dynamic
    \begin{equation*}
    	\tau_{i}(t+1) = \operatorname{sgn} \left(\frac{1}{N} \sum_{k=1}^{N} \tau_{k} + \nu_{k} \right)
    \end{equation*}
    and summing over $i$, we can rewrite it in terms of the \emph{average activity} $m(t) := \dfrac{1}{N} \sum_{i=1}^{N} \tau_{i}(t)$
    \begin{equation*}
    	m(t+1) = \dfrac{N_{+}}{N}   \operatorname{sgn} ( m(t) + |h| ) + \dfrac{N-N_{+}}{N}  \operatorname{sgn} (m(t) - |h|)
    \end{equation*}
    where $N_{+}$ is the number of positive entries of the pattern $\bm{\xi}$. This modeling fits perfectly with the retrieval task; indeed, the network have three possible behavior:
    \begin{itemize}
    	\item if $m(0) > |h|$, $m(t)=1$ then the network retrieve in 1 step the pattern  $\bm{\xi}$
	\item if $m(0) < -|h|$, $m(t)=-1$ then we have retrieval in 1 step of the pattern  $-\bm{\xi}$
	\item if $|m(0)| < |h|$, $\bm{\sigma}$ tends to a configuration $\bm{\xi}^{0}$ where all neurons have either $+$ or $-$ sign, depending on which sign prevails in the pattern; so $\bm{\xi}^{0} = \operatorname{sgn} \left( \sum_{i=1}^{N} \xi_{i} \right) \bm{1}$.
    \end{itemize}
    In other words, the function $m$ measures the alignment  between  the neuronal configuration $\bm{\sigma}$ and the target pattern $\bm{\xi}$. Notice that this model can reconstruct the pattern only if the initial state $\bm{\sigma}(0)$  is sufficiently close to the associated fixed point, i.e. $m(\bm{\sigma}(0)) > |h|$; otherwise, the network finds two possible configurations associated with the other two fixed points.  \\
    In the case of \textbf{multiple Patterns} $\{\bm{\xi}^{\mu} \}_{\mu=1}^{P}$ with $\bm{\xi}^{\mu} \in \{-1,+1\}^{N}$ and $P>1$, we want a Ljapunov function that have stationary states corresponding to the patterns we want to store. For perfect retrieval, we assume that the patterns are orthogonal, i.e. $\bm{\xi}^{\mu}  \cdot \bm{\xi}^{\nu} = 0 \ \forall \mu \neq \nu$; therefore, if we assume that there are no self-interactions and no external fields, Hebb's rule becomes
    \begin{equation}
    	J_{ik} = \dfrac{1}{N}  (1- \delta_{ik}) \sum_{\mu=1}^{P} \xi_{i}^{\mu} \xi_{k}^{\mu}
    \end{equation}
    thus obtaining the Lyapunov function
    \begin{equation*}
    	L(\bm{\sigma}; N,P,J,\bm{h}=0) = \dfrac{P}{2} - \frac{1}{2} \sum_{\mu=1}^{P} \left[ \dfrac{1}{\sqrt N } \sum_{i=1}^{N} \xi_{i}^{\mu} \sigma_{i} \right]
    \end{equation*}
    such that $L(\bm{\sigma}) \geq L(\pm \bm{\xi}^{\mu}) \ \forall \mu = 1...P$, or rather all patterns (and their opposites) are stationary points of dynamics. This will be the initial setting of the Hopfield model, which is known for its associative memory capacity with $P >1$ patterns.
    
        \subsection{Neural Processes as Markov chains}
        Now, we want to analyse the previously dynamics in probabilistic terms using \emph{Markov Chains}. In this section, we will only state the definitions and results that are of interest to us for the purpose of translating neural dynamics in terms of Markov chains. 
        \subsubsection{Brief Overview of Markov Chains}
        \begin{definition}
        	A discrete time Markov chain at values in a discrete set $\mathcal{S}$ is a sequence of random variables $X = \{X_{t}\}_{t \geq 1}$ such that
        \begin{equation*}
        		\mathbb{P} (X_{t+1} = j | X_{t} =i, X_{t-1}=i_{t-1}...,X_{0} = i_{0}) = \mathbb{P} (X_{t+1} = j | X_{t} =i) =: W_{ij}
        \end{equation*}
    i.e. that the probability of finding oneself in state $j$ depends solely on the state at the previous time $i$. 
        \end{definition}
        The probability $W_{ij}$ can be interpreted as the component of a  \emph{stochastic transfer matrix $W$} where
    \begin{equation*}
    	\sum_{j \in \mathcal{S}} W_{ij} = 1 \ \   \forall i  \ \  \text{and} \ \ W_{ij} \in [0,1] . 
    \end{equation*}
    Therefore, the Markov chain is in bi-univocal correspondence with its transfer matrix. Furthermore, we have    \begin{equation*}
    	p_{t}(X=j | X_{0}) = \left( p_{t-1}(X = j | X_{0}) W \right)_{i} = ... = \left(p_{0}(X = j|X_{0}) W^{t} \right)_{i} 
    \end{equation*}
    \begin{definition}
    	A MC is said ergodic if there exists an integer $\tau$ such that for all pairs of states $(i,j) \in \mathcal{S} \times \mathcal{S}$ we have that 
	\begin{equation*}
		p_{t}( X = j | X_{0} = i) > 0 \ \ \forall t > \tau .
	\end{equation*}
    \end{definition}
    \begin{definition}
    	A distribution $p(X|X_{0})$ is invariant if $p(X|X_{0}) W = p(X|X_{0})$.
    \end{definition}
    \begin{theorem}
    	For any ergodic Markov chain $X$, there exists an unique invariant distribution $p_{\infty}(X|X_{0})$ that is the principal left eigenvector of $W$, such that if $\nu(i,t)$ is the number of visits of state $i$ in $t$ steps then $\lim_{t \rightarrow \infty} \dfrac{\nu(i,t)}{t} = p_{\infty}$ .
    \end{theorem}
    \begin{theorem}
    	Let X be an ergodic MC with invariant distribution $p_{\infty}$. Then 
	\begin{equation*}
		p_{t} (X|X_{0}) = p_{0}(X|X_{0})W^{t} \xrightarrow{t \rightarrow \infty}  p_{\infty}(X) 
	\end{equation*}
	independently of the initial distribution $p_{0}(X|X_{0})$.
   \end{theorem}
   \begin{definition}
   	A stochastic matrix $W$ and a measure $p(X)$ are said to be in detailed balance if 
	\begin{equation*}
		p(X=i) W_{ij} = W_{ji} p(X=j) \ \ \forall i,j \in \mathcal{S}.
	\end{equation*}
   \end{definition}
   
   \subsubsection{Translating Neural dynamics in terms of Markov Chains}
   \label{2.4.2}
   	As analysed in section \ref{2.2}, we have that sequential dynamics with noise can be expressed by 
	\begin{equation}
	\label{eq14}
    \left\{
    	\begin{aligned}
        		&\text{choose randomly } i \text{ in } \{1..N\} \\
        		&\mathbb{P} (\bm{\sigma}(t + \Delta t )) = \frac{1}{2} + \frac{1}{2} \sigma_{i}(t + \Delta t ) \operatorname{tanh} ( \beta 	\varphi_{i} (\bm{\sigma}(t))) 
   	 \end{aligned}
    \right.
\end{equation}
	where $\beta = \frac{1}{T}$ and we use the fact that $\operatorname{tanh} ( \sigma_{i} \beta \varphi_{i} (\bm{\sigma}(t))) =  \sigma_{i} \operatorname{tanh} (  \beta \varphi_{i} (\bm{\sigma}(t)))$ since
	$\sigma_{i}$ takes values in  $\{-1,+1\}$ and $\operatorname{tanh} $ is odd. If $\bm{\sigma}(t)$ is given, then equation \ref{eq14} becomes
	\begin{equation*}
	p_{t+1}(\bm{\sigma}) = \frac{1}{N} \sum_{i=1}^{N} \left[ \left( \frac{1}{2} + \frac{1}{2} \sigma_{i} \operatorname{tanh} (\beta \varphi_{i}(\bm{\sigma}(t))) \right) \prod_{j \neq i} \delta_{\sigma_{j},\sigma_{j}(t)} \right]
\end{equation*}
where the production tells us that all neurons remain unaffected except neuron $i$, while the summation tells us that during sequential dynamics, (almost) all neurons are given the opportunity to be flipped. On the other hand, if probability $p_{t}(\bm{\sigma})$ is given,  then we get
       \begin{equation}
       \label{eq15}
            	p_{t+1}(\bm{\sigma}) = \sum_{\bm{\sigma'}} W[\bm{\sigma}; \bm{\sigma'}] p_{t}(\bm{\sigma'})
         \end{equation}
            with the $2^{N} \times 2^{N}$ transfer matrix given by
         \begin{equation*}
    	W[\bm{\sigma}; \bm{\sigma'}] =  \frac{1}{N} \sum_{i=1}^{N} \left[ \left( \frac{1}{2} + \frac{1}{2} \sigma'_{i} \operatorname{tanh} 		(\beta \varphi_{i}(\bm{\sigma'}) \right) \delta_{\bm{\sigma},\bm{\sigma'}}  + \left( \frac{1}{2} - \frac{1}{2} \sigma'_{i} 			\operatorname{tanh} (\beta \varphi_{i}(\bm{\sigma'}) \right) \delta_{F_{i}(\bm{\sigma}),\bm{\sigma'}} \right]
	\end{equation*}
	where $F_{i} (\bm{\sigma}) = (\sigma_{1},...,\sigma_{i-1},-\sigma_{i},\sigma_{i+1},...,\sigma_{N})$ is the \textbf{flipping operator}. Equation \ref{eq15} is the Markov equation corresponding  to the sequential process $\bm{\sigma}(t) \rightarrow \bm{\sigma}(t+1)$.
	\begin{proposition}
		The process described above by $W$ is ergodic. Then there exists a unique stationary distribution $p_{\infty}$ to which it will converge from any initial distributions over states. This distribution is determined by the stationary condition
		\begin{equation*}
			p_{\infty}(\bm{\sigma}) = \sum_{\bm{\sigma'}} W[\bm{\sigma}; \bm{\sigma'}] p_{\infty}(\bm{\sigma'}) \ \ \forall \bm{\sigma} \in \{-1,+1\}^{N} .
		\end{equation*}
	\end{proposition}
	We observe that, to calculate $p_{\infty}(\bm{\sigma})$, we have to solve a system of $2^{N}$ linear equations for $2^{N}$ values of $p_{\infty}$, which would be a very difficult job. This is why we want to impose a strong condition that would simplify the calculations, the \textbf{Detailed Balance} : 
	\begin{equation}
		W[\bm{\sigma}; \bm{\sigma'}] p_{\infty}(\bm{\sigma'}) = W[\bm{\sigma'}; \bm{\sigma}] p_{\infty}(\bm{\sigma}) \ \ \forall \bm{\sigma}, \bm{\sigma'} \in \{-1,+1\}^{N}.
	\end{equation}
	\begin{theorem}
	\label{thm3}
		Let's consider sequential dynamics without self-interactions ($J_{ii} = 0 \ \forall i$). Then the detailed equilibrium is equivalent to the symmetry of the interactions, i.e.
	\begin{equation*}
		J_{ij} = J_{ji} \  \ \forall i,j \iff W[\bm{\sigma}; \bm{\sigma'}] p_{\infty}(\bm{\sigma'}) = W[\bm{\sigma'}; \bm{\sigma}] p_{\infty}(\bm{\sigma}) \ \ \forall \bm{\sigma}, \bm{\sigma'} \in \{-1,+1\}^{N} .
	\end{equation*}		
	Moreover, if detailed balance holds then the equilibrium distribution is given by 
	\begin{equation}
		p_{\infty}(\bm{\sigma}) \propto e^{\frac{- \mathcal{H}(\bm{\sigma})}{T}}
	\end{equation}
	where $\mathcal{H}(\bm{\sigma})$ is the Ljapunov function of the noiseless dynamics given by
	\begin{equation*}
		 \mathcal{H}(\bm{\sigma}) = - \frac{1}{2} \sum_{i,j=1}^{N} \sigma_{i} J_{ij} \sigma_{j}  - \sum_{i=1}^{N} h_{i} \sigma_{i} .
	\end{equation*}
	\end{theorem}
	Notice that $p_{\infty}(\bm{\sigma})$ corresponds to the Boltzmann-Gibbs distribution for $(\bm{\sigma},J,\bm{h})$ .
	\begin{proof}
		Wlog we can suppose $\bm{\sigma'} \neq \bm{\sigma}$; moreover, we assume $\bm{\sigma'}  = F_{i} (\bm{\sigma})$ where $F_{i}$ is the flipping operator. In this case, the DB-condition is equivalent to 
		\begin{equation*}
			\dfrac{  p_{\infty}(\bm{\sigma'}) e^{- \beta \sigma'_{i} \varphi_{i}(\bm{\sigma'})} }{\operatorname{cosh}(\beta  \varphi_{i}(\bm{\sigma'})) } = \dfrac{  p_{\infty}(\bm{\sigma}) e^{- \beta \sigma_{i} \varphi_{i}(\bm{\sigma})} }{\operatorname{cosh}(\beta  \varphi_{i}(\bm{\sigma})) }
		\end{equation*}
		Notice that $\varphi_{i}(\bm{\sigma}) = \varphi_{i}(\bm{\sigma'})$ because of no self-interaction; hence the DB-condition becomes
		\begin{equation*}
			p_{\infty}(\bm{\sigma'}) e^{- \beta \sigma'_{i} \varphi_{i}(\bm{\sigma})} = p_{\infty}(\bm{\sigma}) e^{- \beta \sigma_{i} \varphi_{i}(\bm{\sigma})}
		\end{equation*}
	Recall that the process is ergodic, so $p_{\infty}(\bm{\sigma}) > 0 \ \ \forall \bm{\sigma}$; therefore, we can express the limit distrubution in terms of the exponential function
	\begin{equation*}
		p_{\infty}(\bm{\sigma}) = \operatorname{exp} \left\{ \beta \left( \sum_{k} h_{k} \sigma_{k} + \frac{1}{2} \sum_{k \neq l} \sigma_{k} J_{kl} \sigma_{l} + K(\bm{\sigma}) \right) \right\}
	\end{equation*}
   	where $K(\bm{\sigma})$ is the \emph{implicit term}. Combining the two previous equations, we obtain that the DB condition equals 
	\begin{equation*}
		 \left\{
    	\begin{aligned}
        		&\operatorname{exp}(g_{i}(\bm{\sigma'})) = \operatorname{exp}(g_{i}(\bm{\sigma}))\\
        		&\dfrac{g_{i}(\bm{\sigma})}{\beta} = -\sigma_{i} \varphi_{i} (\bm{\sigma}) + \sum_{k=1}^{N} h_{k} \sigma_{k} + \frac{1}{2} \sum_{k \neq l} J_{kl} \sigma_{k} \sigma_{l} + K(\bm{\sigma})
   	 \end{aligned}
    \right.
	\end{equation*}
	We observe that, by explicating the local field, we obtain
	\begin{equation*}
	\begin{aligned}
		\dfrac{g_{i}(\bm{\sigma})}{\beta}  &= - \sigma_{i} \left( \sum_{k=1}^{N} J_{ik} \sigma_{k} + h_{i} \right) + \sum_{k=1}^{N} h_{k} \sigma_{k} + \frac{1}{2} \sum_{k \neq l} J_{kl} \sigma_{k} \sigma_{l} + K(\bm{\sigma}) \\
		&= \left( -\sigma_{i} h_{i} +  \sum_{k=1}^{N} h_{k} \sigma_{k} \right) + \left(  - \sum_{k=1}^{N} J_{ik} \sigma_{k} \sigma_{i}  + \frac{1}{2} \sum_{k \neq l} J_{kl} \sigma_{k} \sigma_{l} \right) + K(\bm{\sigma}) \\
		&=  \sum_{k \neq i } h_{k} \sigma_{k} + \left( - \sum_{k \neq i} J_{ik} \sigma_{k} \sigma_{i} + \frac{1}{2} \sum_{\substack{k \neq l , k \neq i , l \neq i}} J_{kl} \sigma_{k} \sigma_{l} + \frac{1}{2} \sum_{k \neq i} J_{ki} \sigma_{k} \sigma_{i}  + \frac{1}{2} \sum_{l \neq i} J_{il} \sigma_{i} \sigma_{l}  \right) + K(\bm{\sigma}) \\
		&= \{ \text{terms non involving index $i$} \} + \frac{1}{2} \sum_{k \neq i} (J_{ki} - J_{ik}) \sigma_{i} \sigma_{k} + K(\bm{\sigma})
	\end{aligned}
	\end{equation*}
	In conclusion, these calculations show that the DB condition holds if and only if there exists a function $K(\cdot)$ such that $K(\bm{\sigma'}) - K(\bm{\sigma}) = \sigma_{i} \sum_{k \neq i} (J_{ik} - J_{ki}) \sigma_{k}$. We can now demonstrate the two implications. Let's now assume that the balance condition applies. Thus, we consider $\bm{\sigma} = F_{j} (\bm{\sigma})$ with $j \neq i$ so we obtain
	\begin{equation*}
	\begin{aligned}
		K( F_{i} F_{j} \bm{\sigma} ) - K (  F_{j} \bm{\sigma} ) &=  \sigma_{i} \sum_{k \neq i} (J_{ik} - J_{ki}) F_{j} \sigma_{k} \\
		&=  \left( \sigma_{i} \sum_{k \neq i} (J_{ik} - J_{ki}) \sigma_{k} \right)  - 2 \sigma_{i} (J_{ij} - J_{ji}) \sigma_{j} 
	\end{aligned}
	\end{equation*}
	whence
	\begin{equation*}
		K( F_{i} F_{j} \bm{\sigma} ) - K (  F_{j} \bm{\sigma} ) - K(F_{i} \bm{\sigma}) + K(  \bm{\sigma}) =  - 2 \sigma_{i} (J_{ij} - J_{ji}) \sigma_{j} 
	\end{equation*}
	and since the left-hand member is invariant with respect to permutation $(i,j)$, the right-hand member must necessarily be invariant, i.e. there must be symmetrical interaction. \\
	On the other hand, if we assume symmetrical interaction, then the DB condition is equivalent to the existence of a function $K$ such that $K(F_{i} \bm{\sigma}) - K(\bm{\sigma}) = 0 $ and this is easily verified by taking, for example, constant $ K(\bm{\sigma}) = K$.
	\end{proof}
    
\section{Statistical Mechanics background}
	Statistical mechanics is a branch of Physics that elucidates the collective behavior of macroscopic systems through the analysis of statistical properties at the microscopic level. \\
	 \emph{Spin models} form a foundational framework in statistical mechanics, offering a conceptual lens to investigate the collective behavior of magnetic systems. At their core, these models represent the angular momentum of atomic spins, influencing the material's magnetic properties. Notable among them is the \emph{spin glass model}, introducing disorder for complex behaviors. A specific variant, the \emph{mean-field spin glass}, simplifies the description for tractable analyses. These models illuminate the dynamics of magnetic materials, serving as invaluable tools to unveil the intricate interplay between microscopic spins and macroscopic magnetic phenomena. Within the domain of neural networks, spin models find a unique application in unraveling the intricate dynamics governing these complex systems. By adapting the principles of statistical mechanics to model the interactions between spins, analogous to neurons in a network, researchers can gain valuable insights into the emergent behaviors, phase transitions, and information processing mechanisms within neural networks. These spin models offer a conceptual bridge, allowing us to draw parallels between the collective behavior of spins in magnetic systems and the dynamic interactions of neurons in a network. In this interdisciplinary approach, spin models prove to be versatile tools, shedding light on the nuanced dynamics that define the computational prowess of neural networks. In this section, we provide a concise overview of key statistical mechanics concepts essential for understanding the Hopfield model. We explore principles directly relevant to this neural network, distilling the foundational elements needed to navigate the interplay between statistical mechanics and the Hopfield model's dynamics. Our aim is to offer a focused and accessible entry into the world of statistical mechanics tailored to the study of neural networks.
	 In general, we have a mean-field spin model with $\bm{\sigma} \in \{-1,+1\}^{N}$, $J \in \mathbb{R}^{N \times N}$ symmetric and the \textbf{Hamiltonian}
	 \begin{equation}
		 \mathcal{H}(\bm{\sigma};N,J,\bm{h}) = - \frac{1}{2} \sum_{i,j=1}^{N} \sigma_{i} J_{ij} \sigma_{j}  - \sum_{i=1}^{N} h_{i} \sigma_{i} 
	\end{equation}
	although our focus will be on the probability distribution
	\begin{equation}
		\rho ( \bm{\sigma} ; \beta,N,J,\bm{h}) = \dfrac{\operatorname{exp} \left( - \beta  \mathcal{H}(\bm{\sigma};N,J,\bm{h}) \right)}{Z_{\beta,N,J,\bm{\sigma}}}
	\end{equation}
	where 
	\begin{equation}
		Z_{\beta,N,J,\bm{h}} = \sum_{\bm{\sigma'}} \exp \left( - \beta  \mathcal{H}(\bm{\sigma'};N,J,\bm{h}) \right)
	\end{equation}
	is the \textbf{Partition function}. We observe that we are exactly in the context described by Theorem \ref{thm3} and this shows the close connection between associative neural network models and those of statistical mechanics. The level of description of statistical mechanics is the \textbf{Mesoscopic level}. For each \emph{mesoscopic state} $i \in \mathcal{T}$, we consider its \emph{energy} $E_{i}$; then, a \emph{thermodynamic state} of the system is described by statistical set $ \{ \rho_{i} \}_{i \in \mathcal{T}} $ interpreted as a probability distribution over the set of states $\mathcal{T}$. Let us now give some definitions. 
	 \begin{definition}
	 	For all $j = 1...|\mathcal{T}|$ we call a pure state $\rho^{(j)}$ if $\rho_{i}^{(j)} = \delta_{ij} \ \ \forall i \in \mathcal{T}$. 
	 \end{definition}
    Notice that the set $\mathcal{S}$ of thermodynamic state is a simplex, so each state $\rho \in \mathcal{S}$ can be expressed as a combination of pure states.
    \begin{definition}
     We define the following functions  where $k_{B}$ is the Boltzmann constant :
     \begin{enumerate}[label=(\roman*)]
 		 \item Internal Energy $U(E,\rho) = \sum_{i \in \mathcal{T}} \ \rho_{i} E_{i}$
		  \item Gibbs Entropy $S(\rho) = - k_{B}   \sum_{i \in \mathcal{T}}  \ \rho_{i} \operatorname{log} (\rho_{i})$
  		\item \textbf{Free Energy} $F(E,\rho,T) = U(E,\rho) - T  S(\rho) $
   \end{enumerate}
    \end{definition}
    \begin{definition}
    	We define a thermodynamic equilibrium,or \textbf{Boltzmann-Gibbs distribution}, the state $ \bar{\rho}$ that minimizes the free energy $F$.
	\end{definition}
	\begin{theorem}
    It holds that
    \begin{equation*}
        \left\{
        \begin{aligned}
            &\bar{\rho}_{i} = \dfrac{e^{-\frac{E_{i}}{k_{B}T}}}{Z} \\
            &Z = \sum_{i \in \mathcal{T}} e^{-\frac{E_{i}}{k_{B}T}} = e^{-\frac{\bar{F}(E,T)}{k_{B}T}} \\
            &\bar{F}(E,T) = \operatorname{inf}_{\rho \in \mathcal{S}} F(E,\rho,T) = F(E,\bar{\rho},T) \ .
        \end{aligned}
        \right.
    \end{equation*}
\end{theorem}
Before moving on to the study of the simplest neural network model, i.e. the Curie Weiss model, let us define other functions that will be studied next.
 \begin{definition}
     We define the following functions :
     \begin{enumerate}[label=(\roman*)]
 		 \item \textbf{Intensive Free Energy} $f_{N,\beta,J,\bm{h}} = - \frac{T}{N} log Z_{\beta,N,J,\bm{h}}$
		  \item Intensive Pressure $A_{N,\beta,J,\bm{h}} = \frac {1}{N} log Z_{\beta,N,J,\bm{h}} $
  		\item Thermal average of observable $g$ \ \ $\omega_{N,\beta,J,\bm{h}}(g) = \sum_{\bm{\sigma}} g(\bm{\sigma}) \rho_{N,\beta,J,\bm{h}}(\bm{\sigma})$
   \end{enumerate}
    \end{definition}
    Furthermore, we will say that one of the defined functions is in the \textbf{Thermodynamic limit} (\textbf{TDL}) if we let N tend to infinity. This concept will be fundamental in the continuation of the article, as we will engage in finding solutions in this context because it would be more simple. However, the existence of this limit will not be analysed, which is a very complicated analytical problem.

\subsection{Curie-Weiss Model}
Let us now turn our attention to a very simple neural network model, the Curie-Weiss model. This model can be described as a system made of $N$ spins $\sigma_{i} \in \{-1,+1\}$  that can interact pairwise and with an external field according to the Hamiltonian
\begin{equation*}
	 \mathcal{H}_{N,J,\bm{h}}(\bm{\sigma}) = -  \sum_{(i,j)} \sigma_{i} J_{ij} \sigma_{j}  - \sum_{i=1}^{N} h_{i} \sigma_{i} 
\end{equation*}
We shall consider a homogenous coupling, i.e. $J_{ij} = \frac{J}{N}$ with $J$ constant, and homogeneous external field, i.e.  $h_{i} = h \ \forall i=1...N$. The \emph{order parameter} of the model is the \textbf{empirical Magnetization} 
\begin{equation}
	m_{N}(\bm{\sigma}) = \dfrac{1}{N} \sum_{i=1}^{N} \sigma_{i}
\end{equation}
which expresses the percentage of spin pointing upwards or downwards. In fact, if $m=1$ then there is all positive spin while if $m=-1$ all negative spin. In fact, we can rewrite the Hamiltonian as a function of $m$ as follows
\begin{equation}
	 \mathcal{H}_{N,J,\bm{h}}(\bm{\sigma}) = - \dfrac {J}{N} \sum_{i>j} \sum_{j} \sigma_{i} \sigma_{j} - h \sum_{i} \sigma_{i} = - \dfrac {NJ}{2} (m_{N}(\bm{\sigma}))^{2} - hN m_{N}(\bm{\sigma}) + \dfrac{J}{2}
\end{equation}
where the last term is the diagonal element that we added to introduce $m$ and therefore we have to subtract it; clearly this term will be left out as it does not affect the minimisation process. This is the reason why we refer to these models as \emph{mean field models}. This allows us to apply the so-called \textbf{Coarse-Graining} process, which consists of simplifying the expression for the partition function and rewriting it as 
\begin{equation*}
\begin{aligned}
	Z_{N} &= \sum_{\bm{\sigma}} e^{-\beta \mathcal{H}_{N}(\bm{\sigma})} = \sum_{m \in \mathcal{M}}  e^{-\beta \mathcal{H}_{N}(m) \Omega_{N}(m)} \\
	&\approx \sum_{m \in \mathcal{M}}  e^{-\beta F_{N}(m)} \approx \sum_{k} e^{-\beta N f_{N}(m_{k}^{*})} \rho_{N}^{(k)}
\end{aligned}
\end{equation*}
where $ \Omega_{N}(m)$ is the number of configurations with magnetization equal to $m$ and $m^{*}$ is the argmin of the function $[U_{N}(m) - T S_{N}(m)]$. The approximations have been made assuming that $\Omega_{N}$ is equal to Gibbs entropy (and not Boltzmann entropy) and that N is very large. Now we want to derive an explicit formula for free energy via \textbf{Laplace's method}.
\begin{theorem}
Let $g \in \mathcal{C}^{2}([a,b])$ and $x_{0} \in (a,b)$ be the only point such that $g(x_{0}) = \max_{x \in [a,b]} g(x)$ and $g''(x_{0}) < 0$. Then
\begin{equation}
\lim_{N \rightarrow \infty} \frac{\int_{a}^{b} e^{N g(x)} \, dx}{e^{N g(x_{0})} \sqrt{\frac{2\pi}{N(-g''(x_{0}))}} }= 1.
\end{equation}
This holds also for $a=-\infty$ and/or $b=+\infty$ and for $x \in \mathbb{R}^{K}$ with $\lim_{N \rightarrow \infty} \dfrac{K}{N}=0$.
\end{theorem}
\begin{corollary}
\label{cor1}
	Let be $g$ a convex function. Then we have
	\begin{equation}
	\label{eq24}
		\lim_{N \rightarrow \infty} -\dfrac{1}{N} \log \int_{-\infty}^{+\infty} dx \ e^{Ng(x)} = \min_{x} g(x).
	\end{equation}
\end{corollary}
We observe that we can use the theorem to rewrite the distribution function as
\begin{equation*}
\begin{aligned}
	Z_{N} &= \int dm \ Z_{N}(m) = \int dm \ e^{-N\beta f_{N}(m)} \\
		&\overset{N >> 1}{\approx} \int dm \ e^{-N \beta f_{N}(m) e^{-\frac{N \beta}{2} f''_{N}(m^{*})(m-m^{*})^{2}}} \\
		&= e^{-N \beta f_{N}(m^{*})} \sqrt{\dfrac{2 \pi}{N \beta f''_{N}(m^{*})}}
\end{aligned}
\end{equation*}
from which we obtain that
\begin{equation*}
	f_{N} = - \dfrac{1}{N \beta} log Z_{N} \overset{N >> 1}{\approx} f_{N}(m^{*}) + \dfrac{log\left( \frac{2 \pi}{\beta f''(m^{*})} \right) - log N}{2 N \beta} \overset{N \to \infty}{\longrightarrow} f(m^{*}).
\end{equation*}
Thus we have obtained that, under TDL, the intensive free energy is simply $f(m)$ calculated at its minimum $m^{*}$, which in turn provides the expectation of the magnetization; so we want an explicit formula for $f$. Notice that
\begin{equation}
	f(m) = - \lim_{N \to \infty} \dfrac{T}{N} log Z_{N}(m) = \frac{1}{2} J m^{2} - h m - T s(m)
\end{equation}
where $s(m) =  \lim_{N \to \infty} \frac{1}{N} log \Omega_{N}(m)$ is the entropy per spin. We will now calculate $s(m)$ using the integral representation of Dirac's delta. Indeed, we have
\begin{equation*}
\begin{aligned}
	\delta(m - m_{N}(\bm{\sigma})) &= \int_{-\infty}^{+\infty} dx \ \frac{N}{2\pi} e^{iNx[m-m_{N}(\bm{\sigma})]} =  \int_{-\infty}^{+\infty} dx \ \frac{N}{2\pi} e^{iNxm-ix\sum_{j}\sigma_{j}} \\
	&=  \int_{-\infty}^{+\infty} dx \ \frac{N}{2\pi} e^{iNxm} \prod_{j=1}^{N} e^{-i\sigma_{i}x}
\end{aligned}
\end{equation*}
thus
\begin{equation*}
\begin{aligned}
	\Omega_{N}(m) &= \sum_{\bm{\sigma}} \delta(m - m_{N}(\bm{\sigma})) = \sum_{\bm{\sigma}}   \int_{-\infty}^{+\infty} dx \ \frac{N}{2\pi} e^{iNxm} \prod_{j=1}^{N} e^{-i\sigma_{i}x} \\
	 &=  \int_{-\infty}^{+\infty} dx \ \frac{N}{2\pi} e^{iNxm}	\prod_{j=1}^{N} \sum_{\sigma_{j}=\pm1} e^{-i\sigma_{i}x} \\
	 &=  \int_{-\infty}^{+\infty} dx \ \frac{N}{2\pi} e^{iNxm} \left( 2 \cos(x) \right)^{N} = \int_{-\infty}^{+\infty} dx \ \frac{N}{2\pi} e^{iNxm} e^{N \log (2cos(x))} . 
\end{aligned}
\end{equation*}
From Corollary \ref{cor1}, we get
\begin{equation*}
	\dfrac{d}{dx} [imx - \log(2\cos(x))] = im - \tan(x) = 0 \iff x^{*} = \atan(im) = i \operatorname{atanh}(m) =  \frac{i}{2} \log \left( \dfrac{1+m}{1-m} \right) 
\end{equation*}
from which
\begin{equation*}
	s(m) = -m \tanh(m) + \log(2 \cosh(\tanh(m))) = -\dfrac{1+m}{2} \log( \frac{1+m}{2}) - \dfrac{1-m}{2} \log( \frac{1-m}{2})
\end{equation*}
where where we put $m = x^{*}$. In conclusion, we obtain the explicit formula for the \textbf{coarse-grained intensive free energy in TDL} given by
\begin{equation}
	f(m) = -\dfrac{J}{2} m^{2} - hm + T \left[ \dfrac{1+m}{2} \log( \frac{1+m}{2}) + \dfrac{1-m}{2} \log( \frac{1-m}{2} ) \right] 
\end{equation}
 \begin{figure}[htpb]
    \centering
    \includegraphics[width=0.8\textwidth]{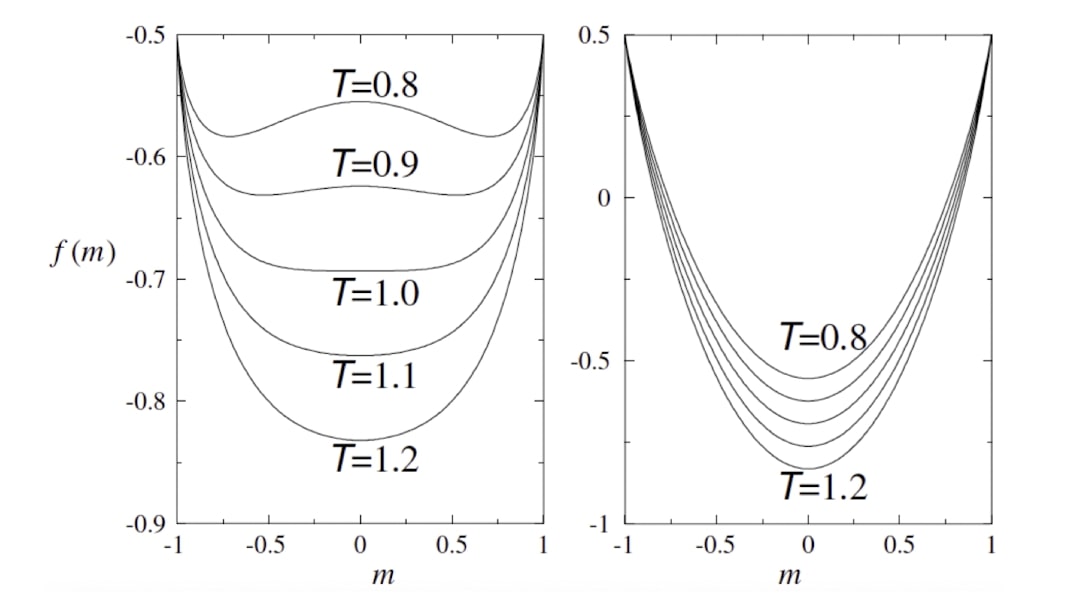}
    \caption{Coarse-grained free energies for the case $J=1$ (left) and $J=-1$ (right)}
    \label{fig5}
  \end{figure}
  
where the first two addends correspond to the energy contribution, while the third addend is the entropy contribution governed by the control parameter $T$. Figure \ref{fig5} shows the plot of $f(m)$ at different temperature levels.  We observe that if $J$ is less than zero, i.e. we are in the presence of \emph{paramagnetic} material, the coarse-grained free energy has a single minimum at $m=0$. Whereas if J is greater than zero, i.e. there is \emph{ferromagnetic} behaviour, then the minimum points vary as the temperature varies. This is a direct consequence of the fact that 
\begin{equation}
	f'(m) = 0 \implies J m^{*} \beta + \beta h = \operatorname{atanh}(m^{*}) = 0 \implies m^{*} = \tanh ( \beta(Jm^{*}+h))
\end{equation}
so the minima correspond to the solutions of the \textbf{Self-Consistency equation} which clearly vary as $\beta$ varies. See Appendix A for further details.

\subsection{Phase Transitions and Ergodicity Breaking}
A \emph{phase transition} occurs when there is a singularity in the free energy or in one of its derivative. This aspect is related to a sharp change in the properties of a substance: for example liquid/gas or paramagnetic/ferromagnetic. It is possible to classify these transitions. For istance, if there is a finite discontinuity in one or more of the prime derivatives, we will say that we are in the presence of a first-order phase transition; if, on the other hand, the prime derivatives are continuous but the second derivatives are discontinuous or finite, then the transition will be of the second order. Phase transitions often involves a \textbf{Symmetry Breaking process}; in fact, the Hamiltonian's system and the equations of the dynamics are invariant under the action of a symmetry group, but the system is not. Tipically, the high-temperature phase contains more symmetries than the low-temperature phase. We denote by $T_{c}$ a critical temperature associated with a phase transition which is sensitively to the interatomic interactions. When $T > T_{c}$ , then the free energy has a single minimum at $m^{*}=0$ and thus the system explores the entire admissible phase space. If $T < T_{c}$, then the free energy has two symmetric minima $\pm m^{*} \neq 0$ and the order of the system is restricted to an appropriate part of the phase space $\Omega^{+}$ or $\Omega^{-}$ with $\Omega = \Omega^{+} \cup  \Omega^{-}$. In the first case the system is ergodic, while in the second case ergodicity is broken. If T=0, the system just moves toward configuration leading to energy minimization so also in this case ergodicity is broken. But in section \ref{2.4.2} we said that the neural dynamics is ergodic; so what is wrong? The fact is that in the TDL $N \to \infty$ the energy barriers cannot be overcome and then ergodicity is broken. It is possible to see this from the point of view of Markov chains; in particular, we denote with $\rho_{j}(t)$ as the probability that the system is in state $j$ at time $t$ and $\bm{\rho} = (\rho_{1}...\rho_{M})$ where $M=2^{N}$ as all the possible states. Then we can express the probability of being at time $t+1$ in state $i$ as \begin{equation*}
	\rho_{i}(t+1) = \sum_{j=1}^{M} W(i \mid j) \rho_{j}(t).
\end{equation*}
with $\bm{W}$ transition matrix. Let us consider its the spectral expansion
\begin{equation*}
	W(i \mid j) = \sum_{k=1}^{M} \lambda_{k} v_{k}^{L}(j) v_{k}^{R}(i)
\end{equation*}
from which we obtain 
\begin{equation*}
	\rho_{i}(t) = \bm{W}^{t} \rho_{i}(0) = \sum_{j=1}^{M} \sum_{k=1}^{M}  \lambda_{k}^{t} v_{k}^{L}(j) v_{k}^{R}(i) \rho_{j}(0).
\end{equation*}
Recall that left and right eigenvectors are orthogonal if they belong to different eigenvalues.
\begin{theorem}
	Let $\bm{W}$ be a stochastic matrix. Then
\begin{enumerate}[label=(\roman*)]
  \item $|\lambda| \leq 1$ for all $\lambda \in \sigma(\bm{W})$
  \item $\lambda_{\text{max}} = 1$ and $\bm{v}^{L} = (1, \ldots, 1)$
\end{enumerate}
\end{theorem}
This implies that 
\begin{equation*}
	\rho_{i}(t) \overset{t \to \infty}{\longrightarrow} \sum_{j=1}^{M} v_{k^{*}}^{R}(i) \rho_{j}(0) = v_{\lambda_{\text{max}}}^{R}(i)
\end{equation*}
then , after a long enough time, memory of the initial state is lost and , for any initial configuration, the asymptotic distribution is reached corresponding to the right eigenvector of the unique largest eigenvalue. Thus, as long as $N$ is finite and $T>0$, the system is ergodic. \\

So how is it possible to break the ergodicity? As $M$ increases, the spectral gap between the maximum and minimum eigenvalue may decrease; then in the limit $M \to \infty$, i.e. in the TDL,  there may be an asymptotic degeneration and ergodicity is broken. For istance, let us denote $v_{1,2}^{R}, v_{1,2}^{L}$ the right and left eigenvector associated to the max eigenvalue with degeneration $2$. Then we obtain
\begin{equation*}
	\rho_{i}(t) \overset{t \to \infty}{\longrightarrow}  \sum_{k=1}^{2} \sum_{j=1}^{M} v_{k}^{L}(j) v_{k}^{R}(i) \rho_{j}(0) = \sum_{k=1}^{2} [v_{k}^{L} \bm{\rho}(0)] v_{k}^{R}
\end{equation*}
and since $\bm{\rho}(0)$ still appears in the equation, the process will lead to different asymptotic trajectories depending on the initial condition. \\

The Curie Weiss model, although we have comprehensive knowledge, is a very simple neural network model. In fact, we have seen how the cost function ( intensive free energy or the Hamiltonian) has, at best, three points of minimum: $m=0$ in correspondence of a completely disordered pattern, $m= \pm m^{*}$ in correspondence of a $\pm \bm{\xi}$ pattern. Thus the model can store with $N$ neurons only one pattern, which is why we will now analyse in detail the much better performing \emph{Hopfield model}.

\section{Hopfield Model}
Let us now analyse a more complex model. We want to obtain a neural network consisting of $N$ neurons, which is characterised by a cost function representing $P$ local minima at the patterns that we want to store in the network. Indeed, let us consider the patterns 
\begin{equation*}
	\bm{\xi}^{\mu} = (\xi_{1}^{\mu},...,\xi_{N}^{\mu}) \in \{-1,+1\}^{N} \ \ \forall \mu = 1...P.
\end{equation*}
Let us consider the standard Hamiltonian of equation \ref{eq11}, where now Hebb's rule is given by
\begin{equation}
	J_{ij} = \sum_{\mu=1}^{P} \dfrac{\xi_{i}^{\mu} \xi_{j}^{\mu}}{N} \ \ \forall i,j=1...N.
\end{equation}
As we did for the Curie-Weiss model, we rewrite the Hamiltonian as a function of magnetization. In this case, however, we do not use the average magnetization but the so-called \textbf{Mattis Magnetization} expressed by
\begin{equation}
	m_{N,\mu} (\bm{\sigma};\bm{\xi}) = \frac{1}{N} \sum_{i=1}^{N} \xi_{i}^{\mu} \sigma_{i} \ \  \forall \mu = 1...P
\end{equation}
i.e. the vector $(m_{N,1},...,m_{N,P})$. We observe that $m$ represents the percentage of equal spins between the $\bm{\sigma}$ configuration and the $\bm{\xi}^{\mu}$ pattern; indeed, if $m_{N,\mu} (\bm{\sigma}) = 1$ this means that the configuration is exactly identical to the stored pattern. If we assume that we have no external field, then we get
\begin{equation}
	\mathcal{H}_{N,P,\bm{\xi}} (\bm{\sigma}) = -\frac {1}{2} \sum_{(i,j)} \sum_{\mu=1}^{P}  \dfrac{\xi_{i}^{\mu} \xi_{j}^{\mu}}{N} \sigma_{i} \sigma_{j} 
	=  -\frac {N}{2}   \sum_{\mu=1}^{P} (m_{N,\mu} (\bm{\sigma}))^{2} + \dfrac{P}{2} 
\end{equation}
where the second addend, which we shall ignore, is linked to the diagonal term. In the remainder of the chapter, we will analyse the solution of the Hopfield model in two separate cases:
\begin{itemize}
    \item \emph{Low-load} case, where $P$ is finite (will be the case we will implement in Python)
    \item \emph{High-load} case, where $P \propto N$ and $\lim_{N \to \infty} \dfrac{P}{N} > 0$
\end{itemize}

\subsection{Solution in the Low-load Regime}
Using Laplace's method, we want to obtain an explicit expression for the free energy. In more detail, we analyse the so-called \textbf{Quenched Intensive Free-energy}
\begin{equation}
	f_{N,\beta,J} ^{Q} = -\dfrac{1}{\beta N} \mathbb{E}[ \log Z_{N,\beta,\bm{\xi}}] 
\end{equation}
where the average is a \emph{quenched average} over possible realisations of patterns given by
\begin{equation*}
	\mathbb{E}[ \ \cdot \ ] = 2^{-NP} \sum_{\bm{\xi} \in \{-1,+1\}^{N \times P} } [ \ \cdot  \ ] 
\end{equation*}
although, to lighten the notation, we will avoid emphasising that all the functions we are now going to calculate are quenched. We observe that the partition function can be written as
\begin{equation*}
\begin{aligned}
	Z_{N,\beta,\bm{\xi}} &= \sum_{\bm{\sigma}} \exp ( \dfrac{\beta}{2N} \sum_{i,j,\mu} \xi_{i}^{\mu} \xi_{j}^{\mu} \sigma_{i}  \sigma_{j}) = \sum_{\bm{\sigma}} \int \left[ \prod_{\mu=1}^{P} dm_{\mu} \delta(m_{\mu} - \sum_{i} \dfrac{\xi_{i}^{\mu} \sigma_{i}}{N}) \right] \exp( \dfrac{\beta N}{2} \sum_{\mu=1}^{P} m_{\mu}^{2}) \\
	&= \sum_{\bm{\sigma}} \int \int \left(  \prod_{\mu=1}^{P} dm_{\mu} \dfrac{N}{2\pi} d\tilde{m}_{\mu} \right) \exp \left( iN \sum_{\mu} \tilde{m}_{\mu} m_{\mu} - i \sum_{j,\mu} \tilde{m}_{\mu} \xi_{j}^{\mu} \sigma_{j} +  \dfrac{\beta N}{2} \sum_{\mu} m_{\mu}^{2} \right)  \\ 
	&=  \int \int \left(  \prod_{\mu=1}^{P} dm_{\mu} \dfrac{N}{2\pi} d\tilde{m}_{\mu} \right) \exp \left[ iN \sum_{\mu} \tilde{m}_{\mu} m_{\mu} + \sum_{j} \log(2 \cos (\sum_{\mu} \tilde{m}_{\mu} \xi_{j}^{\mu})) +  \dfrac{\beta N}{2} \sum_{\mu} m_{\mu}^{2} \right]
\end{aligned}
\end{equation*}
so that the partition function has a linear spin-dependency and can therefore sum directly over configurations. Notice that the extremality conditions in order to use Laplace's method are
\begin{equation*}
    \tilde{m}_{\mu} = i \beta m_{\mu} \ \  \text{with respect to} \ \  m_{\mu} \ \ \text{and} \ \  m_{\mu}^{*} = m_{\mu} = \dfrac{i}{N} \sum_{j} \xi_{j}^{\mu} \tanh(\tilde{m}_{\mu} \xi_{j}^{\mu}) \ \   \text{wrt to} \ \ \tilde{m}_{\mu}
\end{equation*}
from which we obtain
\begin{equation*}
	f_{\beta,\bm{\xi}} = \min_{\bm{m}} \left[ \frac{1}{2} \sum_{\mu} m_{\mu}^{2} - \frac{1}{\beta} \mathbb{E} \left( \log (2 \cosh \beta \sum_{\mu} m_{\mu} \bm{\xi}^{\mu} ) \right) \right]
\end{equation*}
where we use the fact that, for all $\xi_{j} \in \{-1,+1\}^{P}$
\begin{equation*}
	\dfrac{1}{N} \sum_{j} \log(2 \cos (\sum_{\mu} \tilde{m}_{\mu} \xi_{j}^{\mu})) = \dfrac{1}{N} \sum_{j} g(\xi_{j}) \overset{N \to \infty}{\longrightarrow} \mathbb{E}[g(\xi_{j})] .
\end{equation*}
\begin{theorem}
The quenched free energy of the Hopfield model in the thermodynamic limit and in low-load regime is 
\begin{equation}
	f_{\beta}^{Q} = \frac{1}{2} \sum_{\mu} (m_{\mu}^{*})^{2} - \frac{1}{\beta} \mathbb{E} \left[ \log 2 \cosh( \beta \sum_{\mu} m_{\mu}^{*} \bm{\xi}^{\mu}) \right]
\end{equation}
where the Mattis magnetization satisfy the self-consistency equations
\begin{equation}
\label{eq33}
	m_{\mu}^{*} = \mathbb{E} \left[ \bm{\xi}^{\mu} \tanh( \beta \sum_{\nu} m_{\nu}^{*} \bm{\xi}^{\nu}) \right] .
\end{equation}
\end{theorem}
We can see from equation \ref{eq33} how the Curie-Weiss model is a special case of the Hopfield model. In fact, if we assume that only one pattern is the candidate to be retrieved, for istance $\bm{\xi}^{1}$, then we have $m_{1} \neq 0$ while $m_{\mu} = 0 \ \ \forall \mu > 1$ so
\begin{equation*}
	m_{\mu}^{*} = \mathbb{E} \left[ \bm{\xi}^{\mu} \tanh( \beta m_{1}^{*} \bm{\xi}^{1}) \right] =  \tanh( \beta m_{1}^{*})  \ \mathbb{E} \left[ \bm{\xi}^{\mu}  \bm{\xi}^{1} \right] =  \tanh( \beta m_{1}^{*}) \delta_{\mu,1}
\end{equation*}
which is exactly the Curie-Weiss law. Also in this case, there exists a critical  $\beta_{c}= 1$ such that if $\beta < \beta_{c}$ then there is a paramagnetic behavior and if $\beta > \beta_{c}$ then there is a ferromagnetic behavior. The most important difference between the Curie-weiss model and the Hopfield model in low-load regime is the possibility of encountering the system in a \textbf{Spurious state}: it can be interpreted as a system error during the retrieval process, where $\bm{m}^{(n)} = m_{n}(1...1,0...0)$ i.e. the magnetization is a vector with the first $n$ components equal to $1$ and the remaining equal to $0$. This solution is compatible with equation \ref{eq33}; in fact, if $\mu > n$, then the operator $\mathbb{E}$ factorizes because the argument of $\tanh(\cdot)$ is independent of $\xi^{\mu}$ so $\mathbb{E}[\xi^{\mu}] \cdot \mathbb{E}[\tanh \beta \sum_{\nu} m_{\nu}^{(n)} \xi^{\nu}] = 0$; while if $\mu \leq n$, then the equation has non-zero solution for $\beta > 1$. However, the solution associated with a spurious state is simply a free energy extremal point and not a global minimum point. In more detail, deriving an explicit expression of the Hessiana derivative of $f_{\beta}$, we obtain that if $n$ is odd then there exists $T_{c}^{(n)}$ such that for $T<T_{c}^{(n)}$ the spurious states $\bm{m}^{(n)}$ are local minima, while for n even there are saddle points. 

\subsection{Signal-to-noise Analysis}
We now ask whether Hebb's rule stabilises the stored pattern $\bm{\xi}^{\mu}$. For example, will the configuration $\bm{\sigma}$  given by $\sigma_{i} = \xi_{i}^{\mu} \ \ \forall i=1...N$ be dynamically stable? If we assume the absence of external fields and noise, the stability condition is equivalent to saying that $\sigma_{i} \varphi_{i}(\bm{\sigma}) > 0 \ \ \forall i=1...N$. In this way, the configuration does not vary during neural dynamics according to eq.\ref{eq5}, hence the pattern is a fixed point for the model. For istance, let's consider $\bm{\sigma} = \bm{\xi}^{1}$; recall that 
\begin{equation*}
	\varphi_i(\bm{\sigma}) = \sum_{j=1}^N J_{ij} \sigma_{j}(t) = \dfrac{1}{N} \sum_{j \neq i} \sum_{\mu=1}^{P} \xi_{i}^{\mu} \xi_{j}^{\mu} \sigma_{j}
\end{equation*}
thus we obtain
\begin{equation}
\label{eq34}
	\sigma_{i}  \varphi_{i}(\bm{\sigma}) = \xi_{i}^{1}  \varphi_{i}(\bm{\xi}^{1}) = \dfrac{N-1}{N} + \dfrac{1}{N} \sum_{j \neq i} \sum_{\mu=2}^{P}  \xi_{i}^{1}   \xi_{i}^{\mu}   \xi_{j}^{\mu}  \xi_{j}^{1} 
\end{equation}
where the first addend is associated with the \emph{signal} and the second with \emph{noise}.
We observe that the signal term is equal to $1$ in the TDL, while the noise term for very large $N$, denoted by $R$, verifies
\begin{equation*}
	|R| \overset{N >> 1}{\approx} \sqrt{\dfrac{P}{N}}
\end{equation*}
so if $P$ is finite and $N$ very large, the noise becomes negligible in relation to the signal and thus each pattern is effectively a fixed point. This result remains valid even if a finite fraction of spins is flipped away at random from one of the patterns. Although our aim was to build a model and its cost function in such a way as to have minima in correspondence with patterns, the non-linearity of the dynamics means that additional attractors are created. Indeed, let's consider a configuration given by
\begin{equation*}
	\sigma_{i}^{(3)} = \operatorname{sgn}( \xi_{i}^{1} +  \xi_{i}^{2} +  \xi_{i}^{3}) \ \ \forall i =1...N
\end{equation*}
whose Mattis magnetization equals 
\begin{equation*}
	m_{\mu}^{(3)} = \dfrac{1}{N} \sum_{i=1}^{N}  \operatorname{sgn}( \xi_{i}^{1} +  \xi_{i}^{2} +  \xi_{i}^{3}) \xi_{i}^{\mu}  \overset{N >> 1}{\approx}  \mathbb{E} [  \operatorname{sgn}( \xi_{i}^{1} +  \xi_{i}^{2} +  \xi_{i}^{3}) \xi_{i}^{\mu}] = \frac{1}{2} \ \ \forall \mu = 1,2,3
\end{equation*}
and $m_{\mu}^{(3)} = 0 \ \ \forall \mu > 3$. Then we have 
\begin{equation*}
	\sigma_{i}^{(3)} \varphi_{i}(\bm{\sigma}^{(3)}) = \sigma_{i}^{(3)}  \sum_{\mu=1}^{3} m_{\mu}^{(3)} \xi_{i}^{\mu} + \dfrac{1}{N} \sum_{j=1}^{N} \sum_{\mu > 3} \sigma_{i}^{(3)} \xi_{i}^{\mu} \xi_{j}^{\mu} \sigma_{j}^{(3)}
\end{equation*}
and we can establish that this configuration is also stable since, for very large $N$, we have the first term $S = \frac{1}{2} |  \xi_{i}^{1} +  \xi_{i}^{2} +  \xi_{i}^{3}| $  and the second term $R$ such that $R^{2} \approx \frac{P-3}{N}$ so as before, the noise term is negligible. Now, we want to use these results to obtain a \textbf{ Statistical Estimate of the Storage}, i.e. the number of patterns $P_{c} = \max \{ P \ \  \text{s.t. we have retrieval} \} = \alpha N$. \\
We observe that the $S$,$R$ terms in equation \ref{eq34} verify
\begin{equation*}
	S \overset{N \to \infty}{\longrightarrow} 1 \ \text{,} \ R  \overset{N \to \infty}{\longrightarrow} \mathcal{N}(0,\alpha)
\end{equation*}
thus we obtain
\begin{equation*}
	\mathbb{P} [ \sigma_{i} = \xi_{i}^{1} \ \text{stable} ] = \mathbb{P} [  \xi_{i}^{1} \varphi_{i}(\bm{\xi}^{1}) > 0 ] = \mathbb{P} [ R > -1 ] = \frac{1}{2} \left[ 1 + \operatorname{erf}(\frac{1}{2\alpha}) \right].
\end{equation*}
If we assume $\alpha << 1$, then we can approximate the error function as $\operatorname{erf}(x) \overset{x >> 1}{\approx} 1 - \frac{ \exp(-x^{2}) }{\sqrt{ x \pi}}$ thus achieving that 
\begin{equation*}
	\mathbb{P} [ \bm{\sigma} = \bm{\xi}^{1} \ \text{stable} ] \approx \left[ 1 - \sqrt{\frac{\alpha}{2 \pi}} \exp^{- \frac{1}{2 \alpha}}\right]^{N} \approx 1 - N \sqrt{\frac{\alpha}{2 \pi}}  \exp^{- \frac{1}{2 \alpha}} = 1 - N_{\text{err}}.
\end{equation*}
So let us impose the condition $N_{\text{err}} << 1$, which is satisfied for $\alpha = \frac{1}{2 \log N}$. In conclusion, the critical number of patterns to guarantee stability of pattern $\bm{\xi}^{1}$  is
\begin{equation}
	P_{c} = \dfrac{N}{2 \log N}
\end{equation}
and if we want stability also for the other patterns, one reaches $P_{c} = \frac{N}{4 \log N}$. In the particular case of the model implementation that we shall see in Chapter 5, we have $N=513$ and thus $P_{c} \approx 0.1$ \ .

\subsection{Solution in the High-load Regime}
In this section, we will analyse the model in the case of  \emph{High-load}, where $P \propto N$ and 
\begin{equation}
	\alpha :=  \lim_{N \to \infty} \dfrac{P}{N}  > 0.
\end{equation}
One possible approach is to use the so-called \textbf{Replica trick}. This method uses the following identity
\begin{equation*}
	\log(x) = \lim_{n \to 0} \dfrac{x^{n}-1}{n}
\end{equation*}
in order to calculate the quenched pressure
\begin{equation}
\label{eq37}
	A_{\beta}^{Q} = - \beta f_{\beta}^{Q} = \lim_{N \to \infty} \lim_{n \to 0} \dfrac{\mathbb{E} Z_{N,\beta,J}^{n} - 1}{Nn}.
\end{equation}
The trick is to consider distinct replicas $\bm{\sigma}^{(a)}$,$\bm{\sigma}^{(b)}$ that have the same initial distribution. This leads us to introduce the new order parameter called \textbf{overlap}
\begin{equation}
	q_{ab} = \dfrac{1}{N} \sum_{i=1}^{N} \sigma_{i}^{(a)} \sigma_{i}^{(b)}
\end{equation}
which measures the correlation between the two replicas. This method is very efficient for solving the \emph{Sherrington-Kirkpatrick model}, however it is more complicated to apply it to the Hopfield model (see appendix C for more informations about the SK model). \\
An alternative and much more sophisticated approach is the so-called  \textbf{Intepolation technique}. The main idea is to introduce an interpolating pressure $A_{N}(t)$ that recovers the original model for $t = 1$, while for $t = 0$ it corresponds to the pressure of a simpler model analytically addressable. Then, the expression for $A_{N}(t)$ is obtained by exploiting the fundamental theorem of calculus
\begin{equation}
	A_{N,\beta}^{Q} = A_{N}(1) = A_{N}(0) + \int_{0}^{1} \frac{d}{dt} \ A_{N}(t') \ dt' .
\end{equation}
The resolution is based on two starting assumptions. The first consists of the so-called \textbf{Replica Symmetry Ansatz}, in which it is assumed that $q_{ab} = q \ \forall a \neq b$. The second, consists of considering the patterns in the following way: the target pattern $\bm{\xi}^{1}$ is a Rademarcher pattern, while all others $\{ \bm{\xi}^{\mu} \}_{\mu = 2...P}$ are distributed as a standard Gaussian. Recall that the partition function is 
\begin{equation*}
\begin{aligned}
	 Z_{N,\beta,\bm{\xi}} &=  \sum_{\bm{\sigma}}  \exp(-\beta \mathcal{H}_{N,\beta,\bm{xi}}({\bm{\sigma}})  ) 
	&= \sum_{\bm{\sigma}}  \exp( \dfrac{\beta}{2N} \sum_{i,j} \xi_{i}^{1} \xi_{j}^{1} \sigma_{i} \sigma_{j} + \dfrac{\beta}{2N} \sum_{\mu>1} \sum_{i,j} \xi_{i}^{\mu} \xi_{j}^{\mu} \sigma_{i} \sigma_{j})
\end{aligned}
\end{equation*}
and using the identity 
\begin{equation}
	\int dz \exp( -A z^{2} + B z) = \sqrt{\frac{\pi}{A}} \exp(\dfrac{B^{2}}{4A})
\end{equation}
on the second term with $A=\frac{1}{2}$ and $B= \frac{\beta}{N} \left( \sum_{i} \xi_{i}^{\mu} \sigma_{i} \right)^{2}$ we obtain
\begin{equation*}
	Z_{N,\beta,\bm{\xi}} = \sum_{\bm{\sigma}} \int d \mu(z) \exp(  \dfrac{\beta}{2N} \sum_{i,j} \xi_{i}^{1} \xi_{j}^{1} \sigma_{i} \sigma_{j} + \sqrt{\dfrac{\beta}{N}} \sum_{\mu>1} \sum_{i} \xi_{i}^{\mu} \sigma_{i} z_{\mu})
\end{equation*}
where $\mu(z)$ is the Gaussian measure and $z_{\mu}$ is a real variable. The orders parameters shall be the overlap $q_{12}$, the Mattis magnetization $m_{1}$ and also
\begin{equation*}
	r_{12} = \dfrac{1}{P-1} \sum_{\mu} z_{\mu}^{(1)} z_{\mu}^{(2)} \ , \  r_{11} = \dfrac{1}{P-1} \sum_{\mu} z_{\mu}^{(1)} z_{\mu}^{(1)}.
\end{equation*}
Now, we will simply illustrate the results required to arrive at the solution of the Hopfield model, leaving the very tiring calculations to the reader. Let be $t \in \mathbb{R}^{+}$ , $A,B,C,D$ constants to be set a posteriori and $J_{i}, \tilde{J}_{\mu} \sim \mathcal{N}(0,1)$; then the partition function is
\begin{equation}
 \begin{split}
    Z_{N,\beta,\bm{\xi}} (t; J, \tilde{J})&= \sum_{\bm{\sigma}} \int d \mu(z) \exp\left\{ \frac{t}{2} \beta N m_{1}^{2} + \sqrt{\frac{t\beta}{N}} \sum_{\mu,i} \xi_{i}^{\mu} \sigma_{i} z_{\mu} \right. \\
    &+ (1-t)NDm_{1} + (1-t) \frac{C}{2} \sum_{\mu} z_{\mu}^{2} \\
    &+ \left. \sqrt{1-t} A \sum_{i} J_{i} \sigma_{i} + \sqrt{1-t} B \sum_{\mu} \tilde{J}_{\mu} z_{\mu} \right\}
\end{split}
\end{equation}
and the corresponding quenched statistical pressure is
\begin{equation*}
	A_{N}(t) = \dfrac{1}{N} \mathbb{E} \log Z_{N,\beta,\bm{\xi}}(t) \ , \ A(t) = \lim_{N \to \infty} A_{N}(t) .
\end{equation*}
Recall that the $\langle \cdot \rangle$ average of the observable $O$  is
\begin{equation}
\langle O \rangle = \mathbb{E} [\omega_{N,\beta,J} (O)] = \mathbb{E} \left[ \frac{\sum_{\bm{\sigma}} O(\bm{\sigma}) \exp(-\beta \mathcal{H}(\bm{\sigma}))}{\sum_{\bm{\sigma}} \exp(-\beta \mathcal{H}(\bm{\sigma}))} \right] .
\end{equation}
\begin{proposition}
	The derivative of the quenched statistical pressure at finite $N$ is 
	\begin{equation*}
		\dfrac{d}{dt} A_{N}(t) = \dfrac{\beta}{2} \langle m_{1}^{2} \rangle + \dfrac{\beta P}{2N} ( \langle r_{11} \rangle - \langle r_{12} q_{12} \rangle ) - D \langle m_{1} \rangle - \dfrac{A^{2}}{2} (1- \langle q_{12} \rangle ) - \dfrac{\beta^{2}}{2} (  \langle r_{11} \rangle - \langle r_{12}  \rangle ) - \dfrac{C}{2}  \langle r_{11} \rangle
	\end{equation*}
	and in the thermodynamic limit we have 
	\begin{equation*}
	\dfrac{d}{dt} A(t) = \dfrac{\beta^{2}}{2} \overline{m} - \dfrac{\beta \alpha}{2} \overline{r} (1 - \overline{q})	
\end{equation*}
	with some fixed $ \overline{m}, \overline{q}, \overline{r}$ and $A = \beta \alpha  \overline{r} , B = \beta  \overline{q} , C = \beta (1 -  \overline{q}) , D = \beta  \overline{m}$.
\end{proposition}
\begin{theorem}
	In the TDL and under RS assumption, the quenched statistical pressure of the Hopfiel model is 
	\begin{equation}
	\begin{split}
    A_{\beta, \alpha} (\bm{m},q,r ) &= - \dfrac{\beta}{2} \bm{m}^{2}  + \log 2 - \dfrac{\alpha \beta}{2} - \beta \alpha \overline{r} (1 - \overline{q}) + \dfrac{\beta \alpha \overline{q}}{2 (1 - \beta (1 - \overline{q})) } \\
    &\quad - \dfrac{\alpha}{2} \log(1 - \beta (1 - \overline{q})) + \mathbb{E} \left[ \log 2\cosh (\beta \overline{m} + J \sqrt{\beta \alpha \overline{r}}) \right]
\end{split}
	\end{equation}
	where $\overline{m}, \overline{q}, \overline{r}$ fulfill the conditions
	\begin{equation*}
    \left\{
    \begin{aligned}
         \overline{q} &= \mathbb{E} \left[ \tanh^{2} ( \beta  \overline{m} +  J \sqrt{\beta \alpha \overline{r}}) \right] \\
         \overline{m} &= \mathbb{E} \left[ \tanh ( \beta  \overline{m} +  J \sqrt{\beta \alpha \overline{r}}) \right]    \\
         \overline{r} &= \dfrac{\beta \overline{q}}{(1 - \beta(1 - \overline{q}))^{2}}
     \end{aligned}
    \right.
    \end{equation*}	
\end{theorem}
In conclusion, the free energy of the system is given by $f_{\beta, \alpha}(\bm{m},q) = - \frac{A_{\beta, \alpha}(\bm{m},q)}{\beta}$ and the extremality conditions are
\begin{equation}
\begin{aligned}
\label{eq44}
	\bm{m} &= \int d \mu(z) \mathbb{E} \left\{ \bm{\xi} \tanh \left[ \beta \left( \bm{m} \cdot \bm{\xi} + \dfrac{\sqrt{\alpha q}}{ 1 - \beta (1-q)} z \right) \right] \right\} \\
	q &= \int d \mu(z) \mathbb{E} \left\{ \bm{\xi} \tanh^{2} \left[ \beta \left( \bm{m} \cdot \bm{\xi} + \dfrac{\sqrt{\alpha q}}{ 1 - \beta (1-q)} z \right) \right] \right\} .
\end{aligned}
\end{equation}
See Appendix A for some graphic results.

\subsection{Phase Diagram}

 \begin{figure}[htpb]
    \centering
    \includegraphics[width=0.8\textwidth]{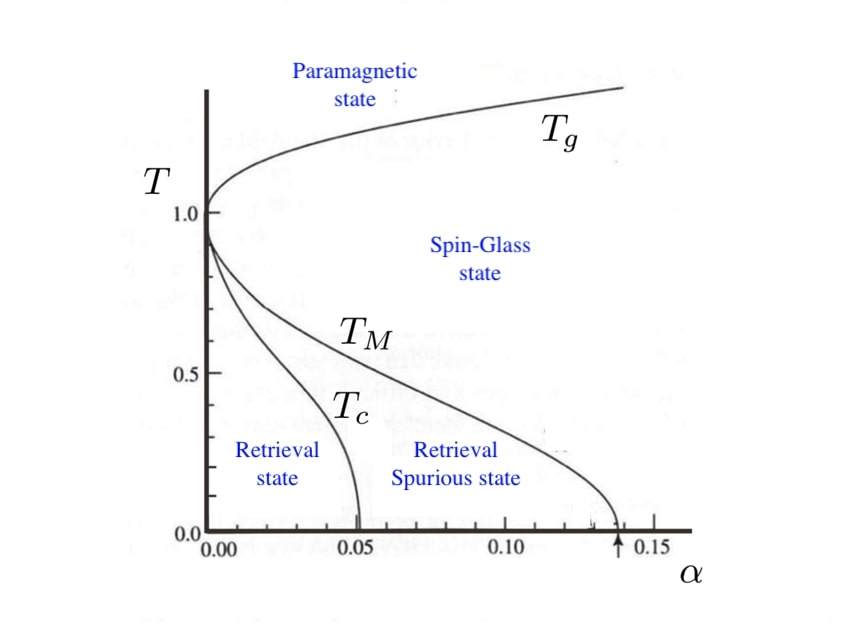}
    \caption{Phase diagram of the Hopfield model at high-load}
  \label{fig6}
  \end{figure}
  
We now have all the tools to analyse the phase diagram of the Hopfield model in detail. As can be seen from the figure \ref{fig6}, four different states can be verified: 
\begin{itemize}
    \item \emph{Retrieval state}, i.e. $\langle m \rangle \neq 0$, $\langle q \rangle \neq 0$.
     \item \emph{Retrieval Spurious state}, i.e. $\langle m \rangle \neq 0$, $\langle q \rangle \neq 0$.
    \item \emph{Spin-glass state}, i.e.$\langle m \rangle = 0$, $\langle q \rangle \neq 0$.
    \item \emph{Paramagnetic state}, i.e.$\langle m \rangle = 0$, $\langle q \rangle = 0$.
\end{itemize}

Firstly, we observe that the retrieval state is characterised by a non-zero magnetization, while the non-retrieval state is characterised by a null, which means that the final configuration is completely random. Secondly, the Retrieval state may be perfect and thus have overlap equal to $1$, or it may be in a Spurious state and retrieve only part of the pattern, thus having the overlap between different replicas be less than $1$. However, this subdivision is more important from the point of view of model applications than from an analytical point of view.. The separation lines between the different states were derived by analytically solving the model for each grid point in order to divide the space $(\alpha,T)$ into the four regions. In more detail, we have that 
\begin{equation}
	T_{g} = 1 + \sqrt{\alpha}
\end{equation}
while $T_{c}$ and $T_{M}$ were derived numerically by comparing the free energies of the pure states and those of the spin-glass states for the same $\alpha$ value.

\section{Audio Retrieval}
 \begin{figure}[htpb]
    \centering
    \includegraphics[width=0.9\textwidth]{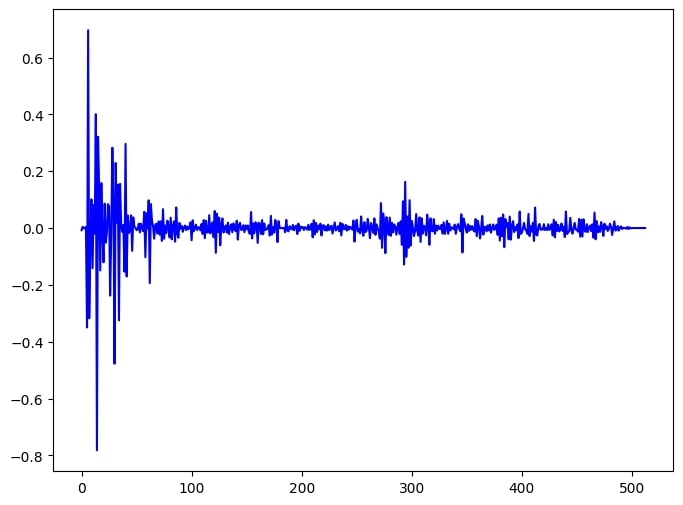}
    \caption{Representation of the average vector of Fourier coefficients of the number $65$}
  \label{fig7}
  \end{figure}

\begin{lstlisting}[style=custompython, caption={Dataset construction}, label=code1,float, breaklines=false, basicstyle=\small\ttfamily]
import numpy as np
import librosa

def fft_signal(audio_path):
    audio, sr = librosa.load(audio_path)
    stft_signal = librosa.stft(audio, n_fft=1024, hop_length=512)
    stft_coeff = np.mean(stft_signal, axis=1)
    return audio, sr, stft_signal, stft_coeff

def audio_importation():
    audio_objects = []
    for i in range(80):
        file_path = f'/Users/silver22/registrazioni/{i}.wav'
        audio, sr, stft_signal, stft_coeff = fft_signal(file_path)
        audio_object = {'audio':audio,'sr':sr,'stft_signal': \
 stft_signal,'stft_coeff':stft_coeff}
        audio_objects.append(audio_object)
    return audio_objects

def audio_binarization(stft_coeff):
    binary = (stft_coeff > 0).astype(int)
    binary = 2 * binary - 1
    return binary
\end{lstlisting}

\begin{lstlisting}[style=custompython, caption={Hopfield Network implementation},label=code2,float, breaklines=false, basicstyle=\small\ttfamily]
import numpy as np
from tqdm import tqdm
import random

class HopfieldNetwork(object):    
  
    def train_weights(self, train_data):
        print("Start to train weights...")
        self.num_neuron = train_data.shape[1]
        self.num_patterns = train_data.shape[0]
        self.patterns = train_data
        J = np.zeros((self.num_neuron, self.num_neuron))
        for i in tqdm(range(0, self.num_neuron)):
            for j in range(i + 1, self.num_neuron):
                for mu in range(0, self.num_patterns):
                    J[i, j] += train_data[mu, i] * \
  train_data[mu, j]
        J = (J + J.T) / self.num_neuron
        self.J = J 
    
    def predict(self, test_data, temperature):
        sigma = test_data.copy().T
        sigma = sigma.T
        N=self.num_neuron 
        K=self.num_patterns 
        alpha = K/N
        T = temperature
        beta = 1.0 / T
        MCstat_step=50
        MCrelax_step=1
        magn_mattis_matrix = np.zeros((self.num_patterns, \ 
  MCstat_step))
        for stat in range(0, MCstat_step):
            for step in range(0,MCrelax_step):
                for i in range(0,N):
                    k = np.random.randint(0, N)  
                    deltaE=2*sigma[k]*np.dot(sigma,self.J[:,k]) 
                    ratio=np.exp(-beta*deltaE)
                    gamma=np.minimum(ratio,1)
                    if np.any(random.uniform(0,1) < gamma):
                        sigma[k] = -sigma[k]  #flipping
            for mu in range(0,self.num_patterns):
                magn_mattis_matrix[mu,stat]= np.dot(sigma, \ 
 self.patterns[mu,:])/N
        predicted = sigma
        return predicted,magn_mattis_matrix
   \end{lstlisting}

A practical application of the Hopfield model to a real dataset is now illustrated. The dataset consists of $81$ voice recordings, in which all numbers from $0$ to $80$ are said. The format of these recordings is ``.wav`` and the Code \ref{code1} in \textbf{Python} was used to transform voice patterns into binary patterns for the network to store. We note that the \emph{Fourier transform} was used to reduce the dimensionality of the audio data; the parameters used (n-fft,hop-length) were set to ($1024,512$) to maintain good audio quality and are powers of two because this provides an enormous computational advantage. Subsequently, an average was taken over each time instant so that a numerical vector of length $513$ representing our audio data could be obtained. In the figure \ref{fig7} the vector of the number $65$ is represented. The second function of the algorithm creates a data structure in which the true audio signal with its sampling rate, matrix and mean vector of the Fourier transform is saved for each audio date. We can clearly see that the vector of coefficients is centred in $0$, so the third function takes care of transforming the vector into a pattern with values in $+1,-1$ so that it can be handled by the Hopfield network. This binarization method, although very crude, turns out to be sufficient for audio retrieval. However, the patterns do not turn out to be orthogonal: in fact, on average, each pattern has $50\%$ of the components equal to the other patterns. Let us now move on to the construction of the model. As we can see from Code \ref{code2}, the Hopfield model can be implemented with two simple functions: the first (train-weights) takes as input a matrix containing the patterns to be stored as rows, and creates the $J$ matrix of synaptic weights according to Hebb's rule. The second (predict), simulates the process of equation \ref{eq10} using \emph{Monte Carlo simulations} (see appendix B for further details). Now, let us analyse the performance of the pattern by doing the following test: we randomly take a pattern from among those stored, and corrupt it by inverting a percentage of components equal to a randomness $r$. We then feed this configuration to the model and compare the Mattis magnetisation relative to that specific pattern. Clearly, we will have that the model has successfully retrieved if m is about $1$, while the prediction has failed if m is less than $0.5$. We will use this heuristic to create graphs that play the same role as the phase diagram, i.e. we will say that we are in a 
\begin{itemize}
    \item  Retrieval state  if $m_{\mu}$ is greater than $0.9$
    \item  Spurious state if $m_{\mu}$ is between 0.6 and $0.9$
    \item non-Retrieval state if $m_{\mu}$ is less than $0.6$.
\end{itemize}
In more detail, we will analyse this process iteratively. That is, at each step a pattern is added to the model and tested; in this way, we want to test the performance of the model as the load changes. In order to make the test as generic as possible, and considering that the voice patterns that make up the dataset have a strong correlation between them (think, for example, of the numbers $21, 22$, etc...), we will perform this pattern addition randomly. At each step, i.e. at a fixed number of patterns, the model is made to work with a temperature ranging from $0.01$ to $2$. In the colourplots shown in the figure, the magnetization value at the end of the Monte Carlo simulation is used. In particular, we show the phase diagram with respect to different randomness values with which to corrupt the test data. 

\begin{figure}[htpb]
    \centering 
    \includegraphics[width=0.9\textwidth]{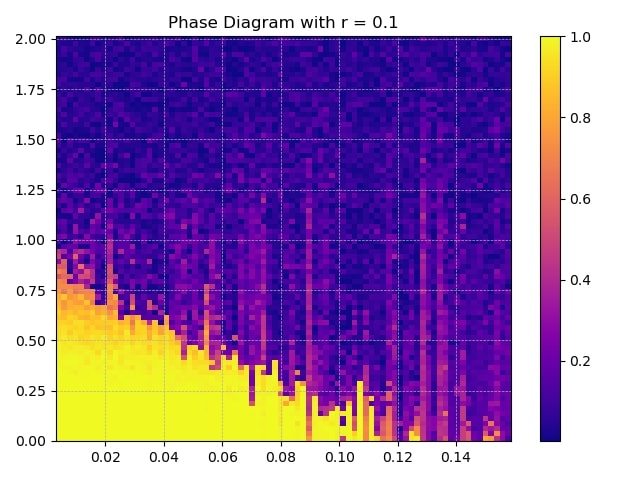}
 \end{figure}
\begin{figure}[htpb]
    \centering 
    \includegraphics[width=0.9\textwidth]{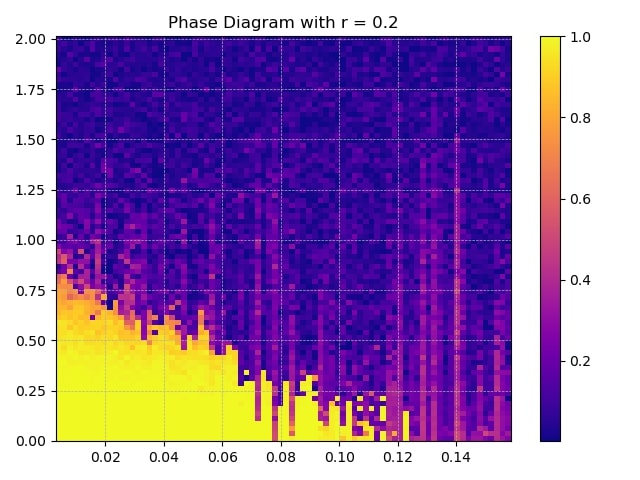}
 \end{figure}
\begin{figure}[htpb]
    \centering 
    \includegraphics[width=0.9\textwidth]{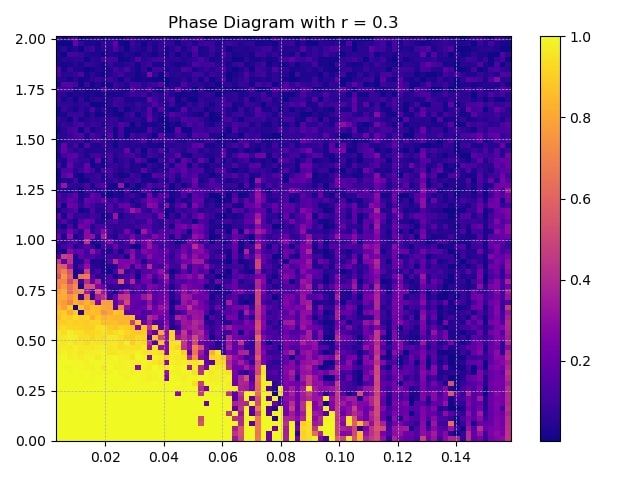} 
    \end{figure}
\begin{figure}[htpb]
    \centering 
    \includegraphics[width=0.9\textwidth]{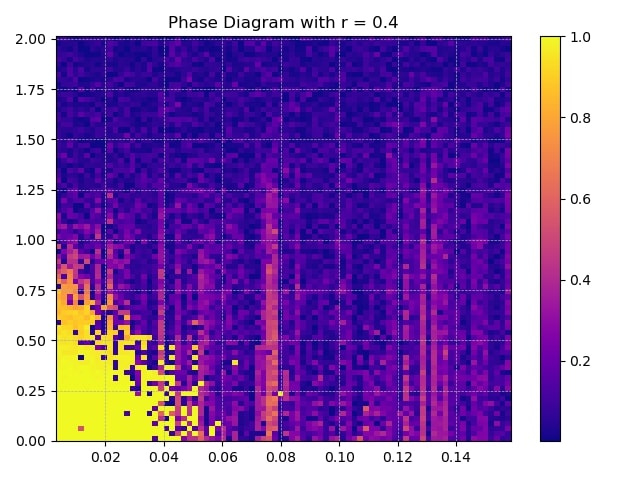}
 \end{figure}
 \begin{figure}[htpb]
    \centering 
    \includegraphics[width=0.9\textwidth]{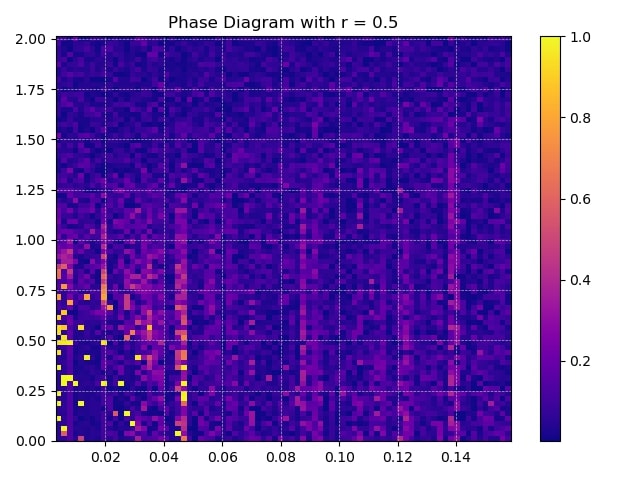}
 \end{figure}

\begin{lstlisting}[style=custompython, caption={Creation of one colourplot with r=0.2},label=code3,float, breaklines=true, basicstyle=\small\ttfamily]
import numpy as np
import matplotlib.pyplot as plt
from tqdm import tqdm
import random
def get_corrupted(self, pattern,r):
        sample_size = int(self.num_neuron*r)
        I = np.random.choice(len(pattern),size = sample_size, \ 
 replace=False)
        corrupted = pattern.copy()
        for i in range(len(I)):
            corrupted[I[i]] = -1*corrupted[I[i]]
        return corrupted
 
audio_object = audio_importation()
num_patterns = 81  
I = np.arange(81)
np.random.shuffle(I)
I_new = I.copy()
patterns_bin_total_list = []
for iter in range(2, num_patterns+1):
    binary_list=[]
    patterns_bin = np.zeros((iter,len(audio_object[0] \ 
['stft_coeff'])))
    for i in range(0,iter):
        stft_coeff = audio_object[I_new[i]]['stft_coeff']
        binary = audiobin2.audio_binarization(stft_coeff)
        patterns_bin[i,:] = binary
    patterns_bin_total_list.append(patterns_bin)
    
model = HopfieldNetwork()
T = T = np.linspace(0.01, 2, 80)
A = np.zeros(80)
for i in range(0,80):
    A[i] = (i+2)/513
magns = np.zeros((len(T),len(A)))
for iter in range(0,num_patterns-1):
    model.train_weights(patterns_bin_total_list[iter])
    print("We're using",iter+2, "patterns")
    rand_test = np.random.choice(range(model.num_patterns))
    randomness = 0.2
    test = patterns_bin[rand_test,:]
    test_corrupted = get_corrupted(test,randomness)
    for t in tqdm(range(0,len(T))) :
        predicted,magnetization = model.predict(test_corrupted, \
 temperature=T[t])
        magns[t,iter] = np.abs(magnetization[rand_test,-1])
        
plt.pcolormesh(A, T, magns, cmap='plasma')
plt.colorbar()
plt.grid(True, linestyle='dashed', linewidth=0.5)
plt.title('Phase Diagram')
plt.tight_layout()
plt.show()
\end{lstlisting}

Below you can find the Code \ref{code3} that calculates the following graphs.

\section{Conclusion}
In this last section we want to draw conclusions about the results obtained in the previous chapter. In particular, we want to highlight the limitations and strengths of the Hopfield model. Firstly, we can say that this model achieves astonishing results compared with its simplicity from both an implementation and a numerical point of view. However, the fact that one must necessarily work with patterns with values in $\{-1,+1\}$ can be a problem in real life, as the binarization of real data may result in a significant loss of the information contained by them. In our case, it was possible to work well with audio data thanks to the Fourier transform, but this may not be the case with other types of data. Secondly, the 'empirical' phase diagrams obtained in Chapter $5$ and the theoretical one obtained in Chapter $4$ show a property of the model that is also common to the biological brain: there is a maximum limit of information that the network can store while maintaining good performance. This can be a great limitation in today's world, since the amount of data is increasingly large and therefore it would be computationally unsustainable to store an NxN matrix with $N>>1$.
Thirdly, we have seen from the various graphs in Chapter $5$ that the model can recognise corrupted patterns up to a certain threshold of randomness; this means that if the network encounters a pattern that is excessively corrupted, it will not be possible to recover the original pattern. This may also be a limitation of the model, because there are currently other neural network models that are able to clean a data item from noise more efficiently (e.g. Autoencoders). 
Following the analysis of these model limitations, further models much more complex than Hopfield's are already being developed, such as Dense Associative Memories. In my humble opinion, these neural network models are fascinating and will continue to be in the spotlight of many researchers around the world, as, unlike almost all Deep Learning models, they have a very detailed mathematical background that allows for a theoretical analysis of the model's capabilities. Having a theory behind a model is very advantageous from a practical as well as an economic point of view, since it would be possible to choose the hyperparameters a priori and not having to test the model using many GPUs. In conclusion, the Hopfield model performs well as an associative memory model for image or audio data and is capable of recovering significantly corrupted patterns.

\section*{Appendix A: Self-Consistency Equations}
\label{A}
Let us analyse the self-consistency equations that characterise the calculation of the extreme points of free energy in the case of the Curie-Weiss model. The equation is 
\begin{equation*}
	m = \tanh[ \beta (Jm + h)]
\end{equation*}
and we show the solutions as the parameters $\beta,J,h$ vary using the fixed point method. Notice that if $\beta > 1$ then $T<1$ and we have more solutions, while if $\beta \leq 1$ we have a unique solution.
\begin{figure}[htpb]
    \centering
    \begin{subfigure}{0.32\textwidth}
        \includegraphics[width=\linewidth]{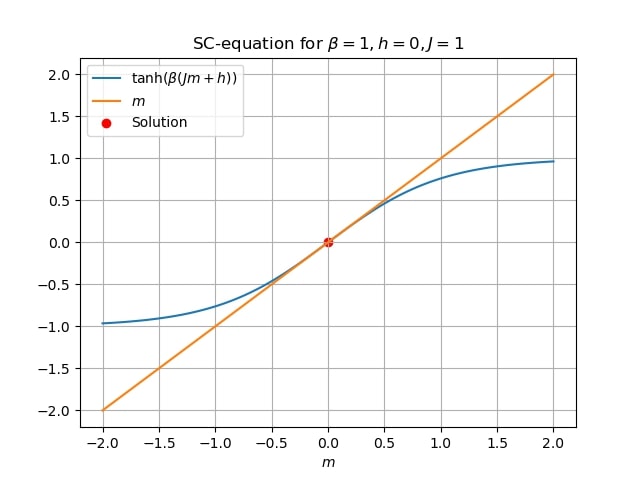}
    \end{subfigure}
    \begin{subfigure}{0.32\textwidth}
        \includegraphics[width=\linewidth]{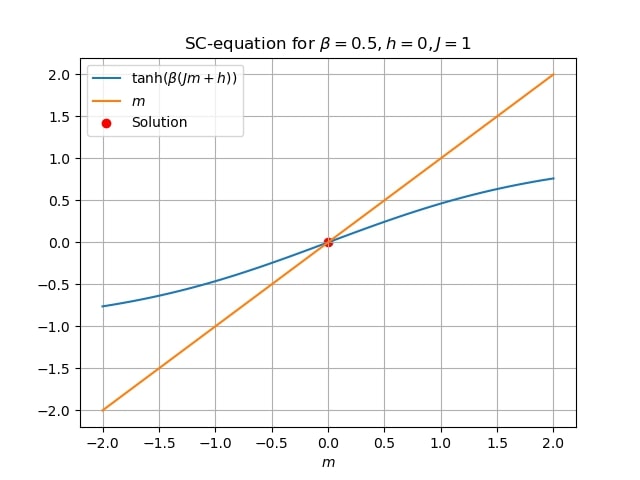}
    \end{subfigure}
    \begin{subfigure}{0.32\textwidth}
        \includegraphics[width=\linewidth]{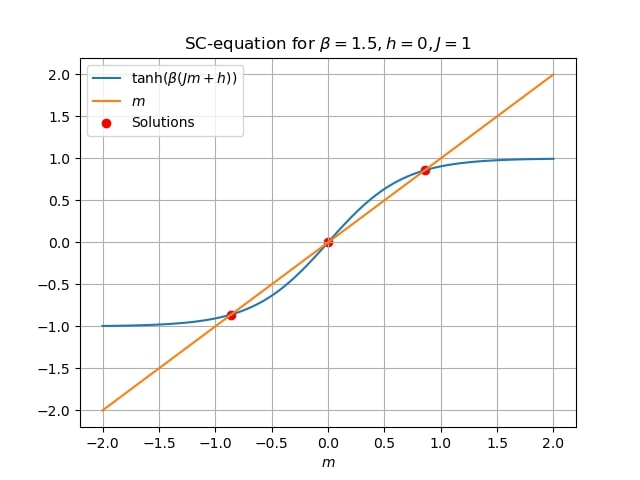}
    \end{subfigure}
\end{figure}

\begin{figure}[htpb]
\centering
    \begin{subfigure}{0.32\textwidth}
        \includegraphics[width=\linewidth]{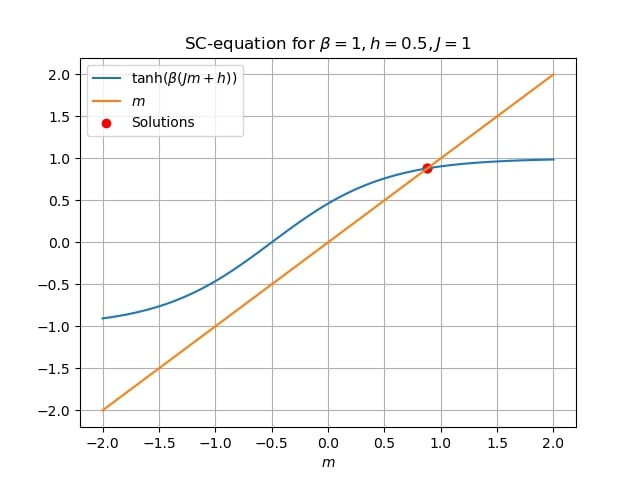}
    \end{subfigure}
    \begin{subfigure}{0.32\textwidth}
        \includegraphics[width=\linewidth]{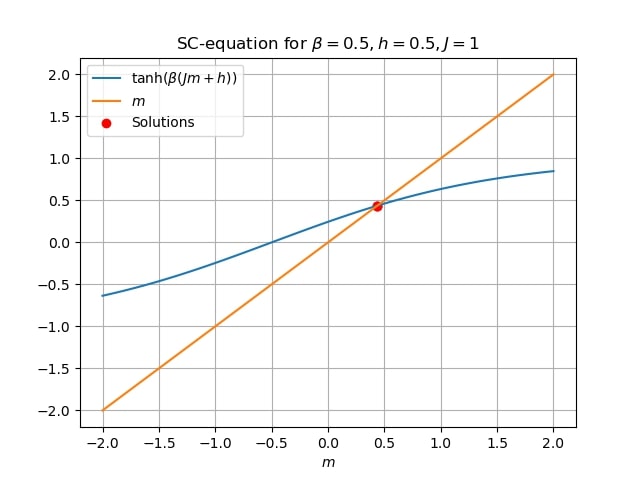}
    \end{subfigure}
    \begin{subfigure}{0.32\textwidth}
        \includegraphics[width=\linewidth]{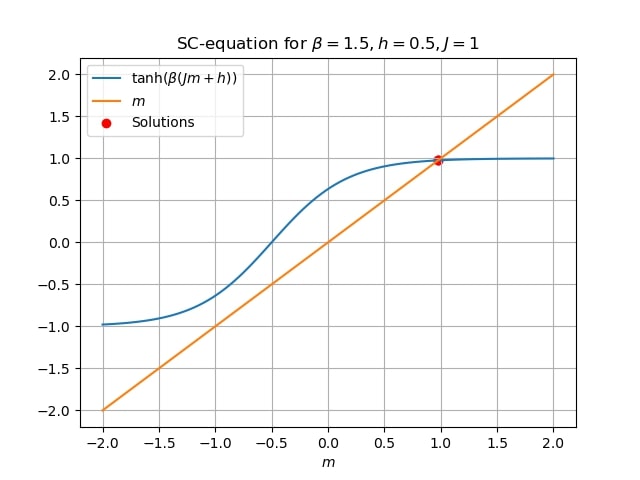}
    \end{subfigure}

\end{figure}

\begin{figure}[htpb]
\centering
    \begin{subfigure}{0.32\textwidth}
        \includegraphics[width=\linewidth]{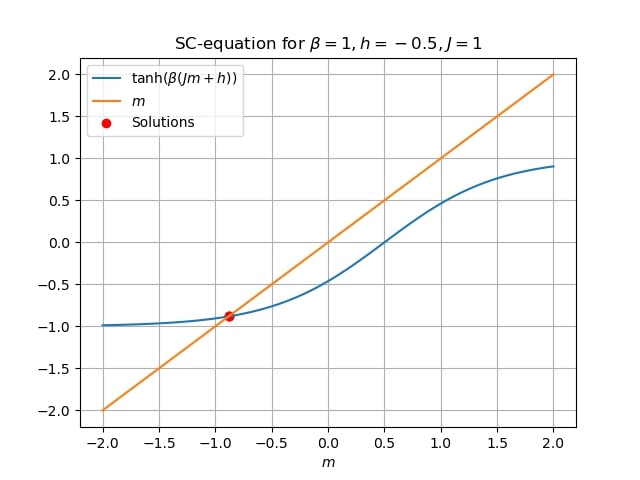}
    \end{subfigure}
    \begin{subfigure}{0.32\textwidth}
        \includegraphics[width=\linewidth]{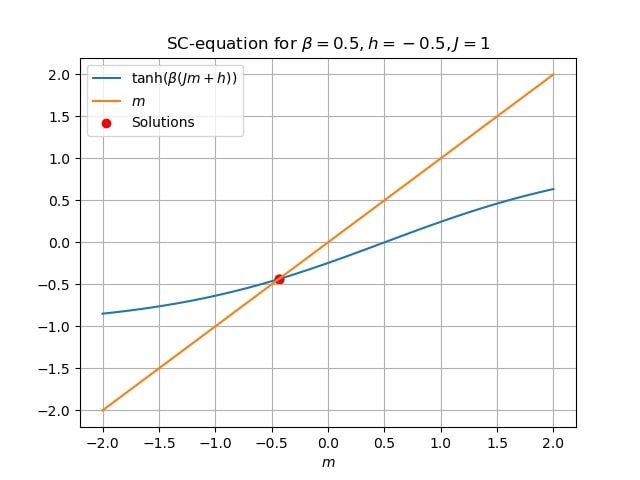}
    \end{subfigure}
    \begin{subfigure}{0.32\textwidth}
        \includegraphics[width=\linewidth]{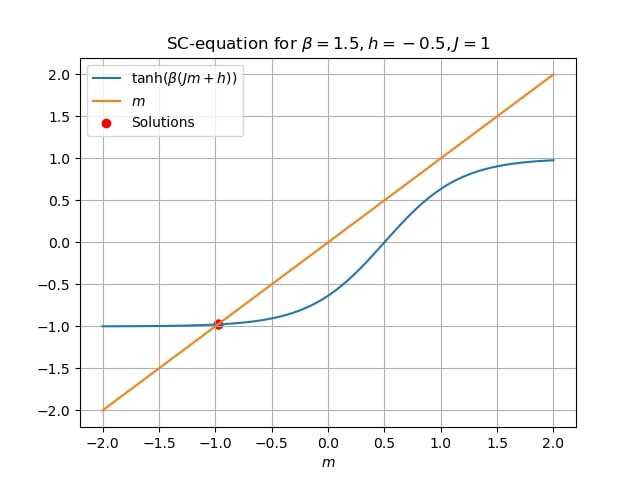}
    \end{subfigure}
\end{figure}
\vspace{\baselineskip}
\vspace{\baselineskip}

Below is also the graphical solution of the self-consistency equations of the high-load hopfield model with the RS assumption using Equation \ref{eq44} from which we can explicitly see the Retrieval zone and the Spin-Glass zone.
\begin{figure}[htpb]
\centering
        \includegraphics[width=0.84\linewidth]{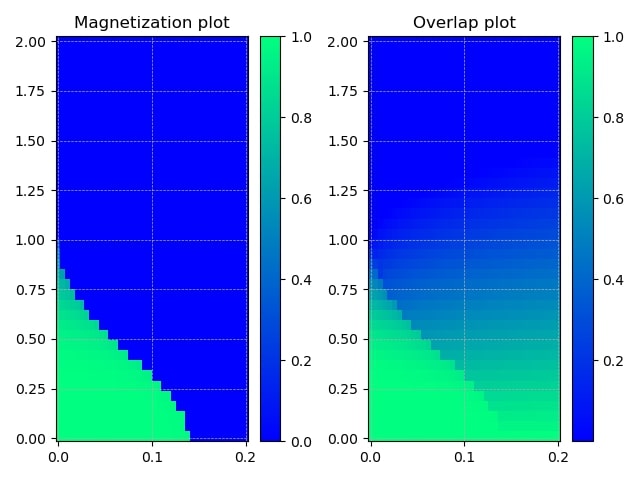}
\end{figure}

\vspace{\baselineskip}
\vspace{\baselineskip}

\section*{Appendix B: Monte Carlo Simulations}

We explain step by step the Monte Carlo simulation used in the \emph{predict} function in Code \ref{code2}. We observe that the following process was used to perform sequential dynamics while optimising the computational cost.. Let's consider the simulation
\vspace{\baselineskip}
\begin{lstlisting}[style=custompython,label=code2, breaklines=true, basicstyle=\small\ttfamily]
        sigma = test_data.copy().T
        sigma = sigma.T
        N=self.num_neuron 
        K=self.num_patterns 
        alpha = K/N
        T = temperature
        beta = 1.0 / T
        MCstat_step=50
        MCrelax_step=1
        for stat in range(0, MCstat_step):
            for step in range(0,MCrelax_step):
                for i in range(0,N):
                    k = np.random.randint(0, N)  
                    deltaE=2*sigma[k]*np.dot(sigma,self.J[:,k]) 
                    ratio=np.exp(-beta*deltaE)
                    gamma=np.minimum(ratio,1)
                    if np.any(random.uniform(0,1) < gamma):
                        sigma[k] = -sigma[k]  
 \end{lstlisting}
 \vspace{\baselineskip}
 The simulation hyper-parameters are $MCstat-step$ and $MCrelax- step$. The former indicates the number of steps to be taken before considering the final value as a statistic. The second, indicates how much you want to relax the process, i.e. how many steps you have to take so that the distribution associated with the network state is close to the Boltzmann-Gibbs distribution, which therefore minimises the free energy. These parameters must be chosen a posteriori, testing the convergence of the magnetisation as the parameters change. In our code, we decided not to relax the process and it turns out that $50$ Monte Carlo steps are sufficient for the model to have convergence in the case of retrieval. The third for loop with respect to variable $i$ corresponds to a single simulation step. The variable $i$ takes values between $1$ and $N$ in such a way that all neurons in the network can be inverted according to the following strategy. At each $i$-step, a neuron is randomly chosen; then the energy contribution $deltaE$ associated with that neuron is calculated. If $deltaE$ is less than $0$, for the minus in the Hamiltonian this gives a positive contribution and thus the state of the neuron will be flipped with probability $1$. If $deltaE$ is greater than $0$, then the probability of finding the neuron in that state is calculated according to the usual formula associated with the Boltzmann-Gibbs distribution; then a random number is extracted in $(0,1)$ and the state of the neuron is inverted only if the number extracted is smaller than the probability calculated previously. This gives a chance to reverse the state of the neuron even though this new configuration does not lead to a decrease in energy. The sigma configuration that will emerge from these three chained cycles will correspond to the new configuration that will hopefully be equal to some fixed point that has been stored in the network.

\section*{Appendix C : Sherrington-Kirkpatrick Model}
The Hamiltonian of the model is
\begin{equation*}
	\mathcal{H}_{N,J}^{SK} = -\dfrac {1}{2 \sqrt{N}} \sum_{i,j} J_{ij} \sigma_{i} \sigma_{j}
\end{equation*}
where the synaptic weights are distributed as a standard Gaussian, i.e. $J_{ij} \sim \mathcal{N}(0,1)$. Notice that the Hamiltonian depends on some random parameters whose probability is supposed to be known, which is why the model is well suited to analysing \emph{disordered systems}. Normalisation with the square root is motivated by the fact that this gives $\langle \mathcal{H} \rangle \propto N$. In the remainder of the appendix, we report the solution of the model using the \emph{Replica Trick} with the \emph{Replica Symmetric Ansatz} (see equation \ref{eq37}). We have
\begin{equation*}
\begin{aligned}
	\mathbb{E}[ Z_{N,\beta,J}^{n}] &= \mathbb{E} \left[ \sum_{\bm{\sigma}^{(1)}} ... \sum_{\bm{\sigma}^{(n)}} \exp(-\beta \sum_{a=1}^{n} \mathcal{H}_{N,\beta,J}({\bm{\sigma}^{(a)}})) \right] \\
	&= \mathbb{E} \left[ \sum_{\bm{\sigma}^{(1)} ... \bm{\sigma}^{(n)}} \exp(\frac{\beta}{2 \sqrt{N}} \sum_{a=1}^{n} \sum_{(i,j)} J_{ij} \sigma_{i}^{(a)} \sigma_{j}^{(a)}) \right] = \\
&=  \sum_{\bm{\sigma}^{(1)}} ... \sum_{\bm{\sigma}^{(n)}} \int \left[ \prod_{i<j} \dfrac{d J_{ij}}{\sqrt{2 \pi}}  \exp(-\dfrac{J_{ij}^{2}}{2}) \exp( \beta \dfrac{J_{ij}}{\sqrt{N}} \sum_{a=1}^{n} \sigma_{i}^{(a)} \sigma_{j}^{(a)} ) \right] =  \\
&= \sum_{\bm{\sigma}^{(1)}} ... \sum_{\bm{\sigma}^{(n)}} \prod_{i<j} \dfrac{\sqrt{2\pi}}{\sqrt{2\pi}} \exp(\dfrac{\beta^{2}}{2N}  \sum_{a=1}^{n}  \sum_{b=1}^{n}  \sigma_{i}^{(a)} \sigma_{j}^{(a)} \sigma_{i}^{(b)} \sigma_{j}^{(b)})
\end{aligned}
\end{equation*}
where in the last step we used the fact that 
\begin{equation*}
	\int dx \exp( -A x^{2} + B x) = \sqrt{\frac{\pi}{A}} \exp(\dfrac{B^{2}}{4A})
\end{equation*}
with $A=\frac{1}{2}$ and $B = \dfrac{\beta}{\sqrt{N}}  \sum_{a=1}^{n} \sigma_{i}^{(a)} \sigma_{j}^{(a)}$. Continuing the rewriting of the $n$-th moment of the partition function, we obtain that 
\begin{equation*}
\begin{aligned}
	\mathbb{E}[ Z_{N,\beta,J}^{n}]  &=  \sum_{\bm{\sigma}^{(1)}} ... \sum_{\bm{\sigma}^{(n)}} \exp( \dfrac{\beta^{2}}{4N} \sum_{i<j}  \sum_{a=1}^{n}  \sum_{b=1}^{n}  \sigma_{i}^{(a)} \sigma_{j}^{(a)} \sigma_{i}^{(b)} \sigma_{j}^{(b)} ) \\
	&= \sum_{\bm{\sigma}^{(1)}} ... \sum_{\bm{\sigma}^{(n)}} \exp( \dfrac{\beta^{2}}{4N} \sum_{i,j}  \sum_{a=1}^{n}  \sum_{b=1}^{n}  \sigma_{i}^{(a)} \sigma_{j}^{(a)} \sigma_{i}^{(b)} \sigma_{j}^{(b)} -  \dfrac{\beta^{2}}{4N} n^{2} N ) \\
	&=  \sum_{\bm{\sigma}^{(1)}} ... \sum_{\bm{\sigma}^{(n)}} \exp{ \dfrac{\beta^{2}}{2N}  \sum_{a \neq b}  \left(\sum_{i} \sigma_{i}^{(a)} \sigma_{i}^{(b)} \right)^{2} +  \dfrac{\beta^{2}}{4N} n N^{2} -  \dfrac{\beta^{2}}{4N} n^{2} N } \\
	&= \exp( \dfrac{\beta^{2}}{4} n (N- n) ) \sum_{\bm{\sigma}^{(1)}} ... \sum_{\bm{\sigma}^{(n)}} \prod_{a<b} \exp{ \dfrac{\beta^{2}}{2N} \left(\sum_{i} \sigma_{i}^{(a)} \sigma_{i}^{(b)} \right)^{2} } 
\end{aligned}
\end{equation*}
and we use the Gaussian result in the other direction with $B = \beta^{2} \sum_{i} \sigma_{i}^{(a)} \sigma_{i}^{(b)}$ , $A = \frac{\beta^{2} N}{2}$ so
\begin{equation*}
	\mathbb{E}[ Z_{N,\beta,J}^{n}] = \exp( \dfrac{\beta^{2}}{4} n (N- n) ) \int  \sum_{\bm{\sigma}^{(1)}} ... \sum_{\bm{\sigma}^{(n)}} \prod_{a<b}  \dfrac{dQ_{ab}}{\sqrt{ \frac{2 \pi}{\beta^{2}}}} \exp( -\frac {\beta^{2}N}{2} Q_{ab}^{2} ) \exp ( \beta^{2} Q_{ab} \sum_{i} \sigma_{i}^{(a)} \sigma_{i}^{(b)} ).
\end{equation*}
Continuing in this way we obtain
\begin{equation*}
\begin{aligned}
	\mathbb{E}[ Z_{N,\beta,J}^{n}] &=  e^{\frac{\beta^{2}}{4} n (N- n)} \int  \prod_{a<b}  \dfrac{dQ_{ab}}{\sqrt{ \frac{2 \pi}{\beta^{2}}}} e^{ -\frac {\beta^{2}N}{2} Q_{ab}^{2} }  \sum_{\bm{\sigma}^{(1)}} ... \sum_{\bm{\sigma}^{(n)}} \prod_{i=1}^{N} e^{\beta^{2} \sum_{a<b} Q_{ab} \sigma_{i}^{(a)} \sigma_{i}^{(b)} } \\
	&= e^{\frac{\beta^{2}}{4} n (N- n)} \int  \prod_{a<b}  \dfrac{dQ_{ab}}{\sqrt{ \frac{2 \pi}{\beta^{2}}}} e^{ -\frac {\beta^{2}N}{2} Q_{ab}^{2} }  \prod_{i=1}^{N} \sum_{\bm{\sigma}^{(1)}} ... \sum_{\bm{\sigma}^{(n)}}  e^{\beta^{2} \sum_{a<b} Q_{ab} \sigma_{i}^{(a)} \sigma_{i}^{(b)} } \\
	&= e^{\frac{\beta^{2}}{4} n (N- n)} \int  \prod_{a<b}  \dfrac{dQ_{ab}}{\sqrt{ \frac{2 \pi}{\beta^{2}}}} e^{ -\frac {\beta^{2}N}{2} Q_{ab}^{2} } \left(  \sum_{\bm{\sigma}^{(1)}} ... \sum_{\bm{\sigma}^{(n)}} e^{\beta^{2} \sum_{a<b} Q_{ab} \sigma_{i}^{(a)}  \sigma_{i}^{(b)} } \right)^{N} \\
	&=  e^{\frac{\beta^{2}}{4} n (N- n)} \int  \left(\prod_{a<b}  \dfrac{dQ_{ab}}{\sqrt{ \frac{2 \pi}{\beta^{2}}}}\right) e^{ -\frac {\beta^{2}N}{2} \sum_{a<b} Q_{ab}^{2} } \exp( N \log  \sum_{\bm{\sigma}^{(1)}} ... \sum_{\bm{\sigma}^{(n)}} e^{\beta^{2} \sum_{a<b} Q_{ab} \sigma_{i}^{(a)}  \sigma_{i}^{(b)} } ) .
\end{aligned}
\end{equation*}
In conclusion, we derived that 
\begin{equation*}
	\mathbb{E}[ Z_{N,\beta,J}^{n}]  = \int \prod_{a<b} \dfrac{dQ_{ab}}{\sqrt{ \frac{2 \pi}{\beta^{2}}}} \exp{ -N \mathcal{A}[Q]}
\end{equation*}
where the argument of the exponential is equal to 
\begin{equation*}
	\mathcal{A}[Q] = -\dfrac{\beta^{2}}{4} n (N-n) + \dfrac{\beta^{2}}{2} \sum_{a<b} Q_{ab}^{2} - \log (\sum_{\bm{\sigma}^{(1)} ... \bm{\sigma}^{(1)}} e^{\beta^{2} \sum_{a<b} Q_{ab} \sigma_{i}^{(a)}  \sigma_{i}^{(b)} } ) .
\end{equation*}
Now, we want to apply the RS-ansatz, i.e. $Q_{ab} = \mathds{1}(a=b) + q \mathds{1}(a \neq b)$. Thus
\begin{equation*}
	\mathcal{A}[Q] = -\dfrac{\beta^{2}}{4} n (N-n) + \dfrac{\beta^{2}}{2} \sum_{a<b} q^{2} - \log (\sum_{\bm{\sigma}^{(1)} ... \bm{\sigma}^{(1)}} e^{\beta^{2} q \sum_{a<b} \sigma_{i}^{(a)}  \sigma_{i}^{(b)} } ) .
\end{equation*}
If we assume commutativity of the limits for $N$,$n$  then we obtain 
\begin{equation*}
A_{\beta}^{Q} = \lim_{n \to 0} \dfrac{1}{n} \lim_{N \to \infty} \dfrac{1}{N} ( \mathbb{E}[ Z_{N,\beta,J}^{n}] - 1) = -  \lim_{n \to 0}  \dfrac{1}{n} \mathcal{A}[Q^{*}]
\end{equation*} 
where we used Laplace's method and $Q^{*} = \operatorname{argmin} \mathcal{A}[Q])$. By calculating the extremal point of $\mathcal{A}$ and doing the limit for $n$, we arrive at the formula
\begin{equation*}
	A_{\beta}^{Q,RS} = \dfrac{\beta^{2}}{4} (1 - q^{RS})^{2} + \log 2 + \log \int d \mu(z) \log \cosh(\beta z \sqrt{q^{RS}})
\end{equation*}
with $q^{RS}$ that satisfies the SC-equation
\begin{equation*}
	q^{RS} = \int d \mu(z) \tanh^{2} ( \beta z \sqrt{q^{RS}}).
\end{equation*}
 \vspace{\baselineskip}
 \vspace{\baselineskip}

\end{document}